\documentclass[11pt]{article}
\usepackage[utf8]{inputenc}
\usepackage[margin=1in]{geometry}
\usepackage{amsmath,amssymb,amsthm,mathtools}
\usepackage{booktabs,longtable,array}
\usepackage{enumitem}
\usepackage{microtype}
\usepackage{xcolor}
\usepackage{comment}
\usepackage{hyperref}
\usepackage{bbm}
\usepackage{tikz-cd}

\newcommand\mc[1]{\mathcal{#1}}

\newcommand{\cC}
{{\mc{C}}}

\newcommand{\cH}{{\mc{H}}}

\newcommand{\cP}{{\mc{P}}}

\newcommand{\cV}{{\mc{V}}}

\newcommand{\cX}{{\mc{X}}}

\newcommand{\cT}{{\mc{T}}}

\newcommand{\cU}{{\mc{U}}}

\newcommand{\PHone}{{\cP_1(\cH)}}
\newcommand{\UH}{{\cU(\cH)}}
\newcommand{\CH}{{\cC(\cH)}}

\newcommand{\F}{{\mathbb F}}
\newcommand{\R}{{\mathbb R}}
\newcommand{\C}{{\mathbb C}}

\newcommand{\RP}{{\mathbb{RP}}}

\newcommand{\one}{{\mathbbm{1}}}
\newcommand{\supp}{{\mathrm{supp}}}

\newcommand{\zz}{{\mathbb{Z}}}
\newcommand{\ra}{{\rightarrow}}

\newcommand{\Sp}{{\mathrm{Sp}}}
\newcommand{\Stab}{\mathrm{Stab}}

\newcommand{\Fp}{\mathbb{F}_p}
\newcommand{\Fq}{\mathbb F_q}

\newcommand{\U}{\mathrm U}
\newcommand{\coker}{\mathrm{coker}}
\newcommand{\im}{\mathrm{im}}
\newcommand{\rank}{\mathrm{rank}}
\newcommand{\tr}{\mathrm{tr}}

\newtheorem{definition}{Definition}
\newtheorem{lemma}{Lemma}

\newtheorem{proposition}{Proposition}
\newtheorem{theorem}{Theorem}

\title{A complete classification of the existence of finite-dimensional quantum solutions to inconsistent linear constraint systems}
\author{Markus Frembs\thanks{\href{mailto:markus.frembs@itp.uni-hannover.de}{markus.frembs@itp.uni-hannover.de}}\\
\footnotesize \textit{Institut f\"ur Theoretische Physik, Leibniz Universit\"at Hannover}\\
\footnotesize \textit{Appelstra\ss e 2, 30167 Hannover, Germany}}
\date{}

\begin{document}
\maketitle

\vspace{-.5cm}

\begin{abstract}
    Linear constraint systems (LCS) provide a compact algebraic language for nonlocal games and state-independent contextuality. In the binary case the Mermin--Peres square shows that inconsistent LCS can be solved by promoting variables to local Pauli operators, whereas a series of no-go results excludes analogous Pauli-, Clifford- and monomial unitary-based constructions at odd prime modulus. We first isolate the reason for this difference: for tensor-product quantum solutions, the obstruction to a classical solution of a LCS over $\mathbb{Z}_n$ decomposes into obstruction classes carried by the individual factors. As a consequence, a genuinely global obstruction cannot arise solely by tensoring locally unobstructed systems. The argument is specific to odd modulus; at even modulus the reordering phase need not vanish, as in Pauli-based examples.
    
    Given this distinction, we instead associate LCS to finite arrangements of rank-one projectors and resolutions of the identity in finite-dimensional vector spaces. Quantum solvability for such LCS is automatic and the search for a quantum-classical gap rests entirely on proving classical unsatisfiability of the underlying modular incidence problem of the arrangement. By relating such LCS with group-valued frame functions, we identify a family of examples in the work of [Harding, Jager, and Smith, Int. J. Theor. Phys. \textbf{44}, 539 (2005)]. To establish the (non)existence of classically unsatisfiable LCS over $\mathbb{Z}_n$ with quantum solutions in all dimensions $d$ and for all $n\in\mathbb{N}$, we further construct explicit examples for the cases $d=n$ prime. Our main result thus positively resolves the existence problem of such LCS beyond the binary case: such systems exist if and only if $d\geq 3$ and $\gcd(n,d)>1$. Moreover, we establish the existence of LCS with a quantum-classical gap over finite fields if and only if $d\geq 3$ and $p\mid d$. Finally, we return to qubits and compare our (symmetric) incidence-based constructions with the (group-)cohomological witnesses of contextuality.
\end{abstract}

\section{Introduction}\label{sec:intro}

Nonlocal games provide an operational paradigm for comparing classical and quantum correlations. Two or more spatially separated players receive questions from a referee and attempt to satisfy a prescribed winning relation without communicating after the questions are received. Entanglement can increase the winning probability, and in a pseudo-telepathy game it can make an otherwise impossible task perfectly winnable. Linear constraint system games represent a class of algebraically rigid examples: a system of linear equations is converted into a two-player game where one player is asked to determine the value of a variable chosen by a referee, and the other player the value of all variables in a given constraint, also given by a referee. Perfect quantum strategies are controlled by operator solutions to such linear constraint systems (LCS) \cite{CleveMittal,CleveLiuSlofstra}. LCS have played a prominent role in the operator-algebraic study of nonlocal correlations. In particular, Slofstra used linear constraint system games to separate tensor-product and commuting-operator models in variants of Tsirelson's problem \cite{SlofstraTsirelson,SlofstraNonclosed,ColadangeloStark2019}; the broader relation between Tsirelson's problem and Connes' embedding problem was developed in \cite{JungeEtAl}, and the latter was ultimately refuted as a consequence of the identity $\mathrm{MIP}^*=\mathrm{RE}$ \cite{MIPstarRE}.

From the single-system point of view, the same operator relations are paradigmatic examples of contextuality: mutually compatible sets of measurements define overlapping classical contexts which do not admit a global value assignment \cite{KochenSpecker,AbramskyBrandenburger}. The binary case is both the best understood and the source of the basic intuition. Cleve and Mittal formalised the corresponding binary constraint system games, and Cleve, Liu and Slofstra identified perfect commuting-operator strategies with representations of the solution group associated with a LCS \cite{CleveMittal,CleveLiuSlofstra}. The Mermin--Peres square and Mermin's pentagram define classically inconsistent LCS over $\mathbb F_2$ with quantum solutions in the form of $2$- and $3$-qubit Pauli operators \cite{Mermin,Peres}. Pauli solutions admit an efficient algebraic description \cite{Arkhipov2012,OkayRaussendorf2020,TrandafirLisonekCabello2022,MullerGiorgetti2025,AbramskyCercelescuConstantin}, but do not exhaust binary operator solutions. Recent work of Zhang, Pan and Liu gives an explicit family of binary constraint systems admitting perfect finite-dimensional strategies outside the Pauli/Clifford setting; their eight-dimensional realisation contains genuinely multi-qubit operators, including a measurement unitarily equivalent to a Toffoli gate, and is therefore not a tensor product of single-qubit Pauli measurements \cite{ZhangPanLiu}. The existence problem for operator solutions is thus genuinely broader than stabiliser theory even for qubits.

At odd prime modulus the familiar local Pauli mechanism breaks down. Qassim and Wallman proved that generalised Pauli solutions cannot produce a quantum-classical satisfiability gap for LCS over $\zz_n$ with $n$ odd and established a phase-commutation obstruction ruling out natural odd-modulus analogues of the Mermin-Peres square and Mermin's pentagram \cite{QassimWallman}. Frembs, Okay and Chung extended these results beyond the Pauli group, first towards diagonal Clifford operators and subsequently to local monomial unitaries, motivated by measurement-based quantum computation \cite{FOC2025,FOC2026}. Slofstra and Zhang in parallel established possibly infinite-dimensional operator solutions for large classes of graph-incidence systems over arbitrary moduli \cite{SlofstraZhang}, yet LCS over $\zz_n$ with $n$ odd and finite-dimensional quantum-classical gaps were not previously identified in the LCS literature, and Chung, Okay and Sikora even conjectured no such examples to exist \cite{ChungOkaySikora2024}. Here, we refute this conjecture. We first observe that the 2005 construction of Harding, Jager and Smith \cite{HardingJagerSmith2005} already implicitly supplies such examples for every $n$ in dimension $2n$, and extend it to a complete and constructive solution in Thm.~\ref{thm: qsols general}. In summary, we prove the following:
\begin{center}
    \textit{LCS with a finite-dimensional quantum-classical gap exist if and only if $d\geq 3$ and $\gcd(n,d)>1$.}
\end{center}
What is more, we also establish the existence of LCS with a quantum-classical gap in finite dimension over arbitrary finite fields $\Fq$ with $q=p^f$ if and only if $d\geq 3$ and $p\mid d$. 

Independent concurrent work in Ref.~\cite{Ciardo2026,vanDobbendeBruynEtAl2026} also establishes the existence of linear constraint systems over arbitrary rings $\zz_n$ with quantum but no classical solutions. Ref.~\cite{Ciardo2026} also builds on the results of Ref.~\cite{HardingJagerSmith2005}, while Ref.~\cite{vanDobbendeBruynEtAl2026} devises a group-theoretic construction, which generalises the Heisenberg-Weyl group, and obtains examples that closely mimic the paradigmatic Mermin-Peres square for two qubits \cite{Mermin,Peres}.\footnote{The results in Ref.~\cite{vanDobbendeBruynEtAl2026} have also been previously announced in e.g. \cite{vdB2025,vdB2026}.\label{fn: prior work}} By comparison, here we not only provide examples for every $n$ but solve the complete $n,d$-classification for the existence of LCS with quantum-classical gaps, and provide explicit constructions in all cases, including the minimal dimension $d=p$.\\

\textbf{Outline.} Sec.~\ref{sec: LCS} recalls the definition of linear constraint systems (LCS).

In Sec.~\ref{sec: locality}, we isolate a structural even--odd distinction which holds for all local unitaries, independent of previous restrictions to Pauli, Clifford or monomial gates: for $n$ odd an obstruction cannot arise only after tensoring locally unobstructed factors. Characteristic two is exceptional because a product of phase-commuting involutions can acquire fourth-root phases. This is precisely the local mechanism used by the Mermin--Peres square and Mermin's pentagram.

Given this distinction, in Sec.~\ref{sec: incidence reduction} we instead consider LCS associated with finite ray-context arrangements (see Def.~\ref{def: canonical LCS}). The operators $\zeta^P_n$ with $\zeta_n=e^{\frac{2\pi i}{n}}$ and $P$ a projector in the arrangement then automatically solve the constraints of the LCS over $\zz_n$ defined by the relations $\sum_{P\in C}x_P=1\mod n$ for every context $C$. This shifts the above existence problem to proving the corresponding incidence relations of an arrangement to be inconsistent. We show that the inconsistency of such a LCS is characterised by a finite incidence cycle, which further simplifies for symmetric arrangements.

Next, we relate the existence of quantum solutions to LCS associated with ray-context arrangements to group-valued frame functions. This allows us to identify a family of LCS with a quantum-classical gap in the work of Harding, Jager and Smith \cite{HardingJagerSmith2005} in Sec.~\ref{sec: qsols}. To fully solve the existence problem, we perform a computational search for symmetric ray-context arrangements in $\C^p$. From this we extract a generic construction for $p\geq 5$ and an explicit $43$ ray-$31$ context arrangement in dimension $3$, which turns out to be a close relative of Peres' and Penrose's Kochen-Specker sets \cite{Peres1991,Penrose1994}, as explicated in App.~\ref{app: d=p=3}; App.~\ref{app: conference LCS} contains additional constructions of LCS with quantum-classical gaps; App.~\ref{app: Pauli comparison} compares and embeds the $n$-qubit Pauli-based constructions of quantum solutions to LCS in the binary case within our approach; Sec.~\ref{sec: conclusion} summarises.

\section{Linear constraint systems and quantum solutions}\label{sec: LCS}

Throughout, $\zz_n:=\zz/n\zz$ for $n\ge2$. Relations between integer representatives are written $a\equiv b\pmod n$, whereas equations stated intrinsically over $\zz_n$ may simply be written as equalities in $\zz_n$. For prime $p$, we identify $\zz_p$ with the finite field $\Fp$, and use $\Fp$ when the field structure is relevant.

Let $2\leq n\in\mathbb{N}$,\footnote{The case $n=1$ corresponds with the trivial ring $\zz_1=\zz/\zz$, which trivialises any (classical solution to an) LCS defined over it; similarly, the $n$-torsion constraint $X^n=\one$ for $n=1$ trivialises quantum solutions to such LCS.} let $A\in\mathrm{Mat}_{M\times N}(\zz_n)$ and $b\in\zz^M_n$. This defines a system of linear equations or a linear constraint system (LCS). A classical solution to an LCS is a vector $x\in\zz^N_n$ satisfying $Ax\equiv b\pmod n$. If no such vector exists, the LCS is called \emph{inconsistent} or \emph{(classically) unsatisfiable}.

Let $\zeta_n=e^{2\pi i/n}$. A \emph{$d$-dimensional quantum solution to a LCS} is a family of unitary operators $X_1,\ldots,X_N\in\UH$ on a Hilbert space $\cH$ of dimension $d=\mathrm{dim}(\cH)$ with unit $\one$, such that:
\begin{align}\label{eq: qsol}
\begin{split}
    X^n_i&=\one\; ,\\
    X_iX_jX^{-1}_iX^{-1}_j&=\one\qquad\text{whenever $A_{ri}\neq 0\neq A_{rj}$ for some $r\in[M]$}\; ,\\
    \prod_iX_i^{A_{ri}}&=\zeta^{b_r}_n\one\qquad\forall r\in[M]\; .
\end{split}
\end{align}
Note that the order in the last constraint is immaterial because the operators in each row commute. For $n=2$ these are the operator-valued binary constraint systems of \cite{CleveMittal}. We say that a LCS has a \emph{(finite-dimensional) quantum-classical gap} if it admits a quantum but no classical solution. LCS with an infinite-dimensional quantum-classical gap are known to exist \cite{SlofstraZhang}, and have played a major role in the resolution of Tsirelson's problem \cite{SlofstraTsirelson,SlofstraNonclosed}. Here, we are interested in the existence of quantum solutions to classically unsatisfiable LCS on finite-dimensional Hilbert spaces. Such solutions yield a perfect strategy for the associated linear system game by sharing a maximally entangled state and using transposed observables on the second subsystem \cite{CleveMittal,CleveLiuSlofstra,ColadangeloStark2019}. LCS with a quantum-classical gap therefore provide purely algebraic and maximal winning strategies in nonlocal tasks. In the binary case, such examples are readily obtained in terms of $n$-qubit Pauli operators \cite{Mermin,Peres,Arkhipov2012,OkayRaussendorf2020,TrandafirLisonekCabello2022,MullerGiorgetti2025,AbramskyCercelescuConstantin}. Yet, a complete solution to the existence problem of finite-dimensional quantum solutions to LCS has remained open up to now.$^{\ref{fn: prior work}}$

We begin with the following easy observation.

\begin{proposition}\label{prop: determinant}
    Let $2\leq n\in\mathbb{N}$ and $d\in\mathbb{N}$. If a classically unsatisfiable LCS over $\zz_n$ has a finite-dimensional quantum solution in dimension $d$, then $\gcd(n,d)>1$.
\end{proposition}

\begin{proof}
    Let $\zeta^{s_i}_n:=\det(X_i)$ with $s_i\in\zz_n$ by $n$-torsion. Taking determinants in the constraint equation of Eq.~(\ref{eq: qsol}) gives $As=db$. Hence, if $\gcd(n,d)=1$, $d$ is invertible and $d^{-1}s$ a classical solution.
\end{proof}

Prop.~\ref{prop: determinant} thus constrains our search for LCS with a quantum-classical gap over $\zz_n$. In our main result, Thm.~\ref{thm: qsols general}, will prove a converse to Prop.~\ref{prop: determinant}: for all pairs $n,d$ with $\gcd(n,d)>1$ and $d\geq 3$ there exists a classically unsatisfiable LCS over $\zz_n$ which admits a quantum solution.

For LCS associated to ray-context arrangements (see Def.~\ref{def: canonical LCS}) this condition has an elementary interpretation: since every context contains $d$ rays, the constant assignment $d^{-1}$ solves every constraint equation. Before turning to this construction, we first rule out an alternative construction in terms of tensor products of local unitaries, as in the case of $n$-qubit Pauli operators.

\section{Quantum solutions from tensor product operators}\label{sec: locality}

The no-go theorems of Ref.~\cite{QassimWallman,FOC2025,FOC2026} rule out a quantum-classical gap for LCS over $\zz_n$ with $n$ odd and quantum solutions consisting of special classes of local unitary operators. Our first result rules out this approach for arbitrary tensor-product unitaries.

Fix a tensor decomposition $\cH=\bigotimes_{k=1}^K\cH_k$ with $d_k=\dim(\cH_k)$ and suppose that
\begin{equation}\label{eq: tensor product operators}
    X_i=\bigotimes_{k=1}^KU_{i,k}\; ,
\end{equation}
for every $i\in[N]$ and $U_{i,k}\in\cU(\cH_k)$ define a quantum solution to some LCS.

Note that the factors are defined up to scalars only. We first choose a convenient rescaling.

\begin{lemma}\label{lm: p-torsion gauge}
    Let $X=\bigotimes_{k=1}^KU_k$ with $X^n=\one$, then there exist unitaries $U'_k\sim U_k$ such that $X=\bigotimes_{k=1}^KU'_k$ and $U'^n_k=\one$ for every $k\in[K]$.
\end{lemma}

\begin{proof}
    From $\bigotimes_{k=1}^KU^n_k=\one$ it follows that $U^n_k=\lambda_k\one$ for every $k$ with $\prod_{k=1}^K\lambda_k=1$. Now, choose $n$-th roots $\mu^n_k=\lambda_k^{-1}$. Then the product of the $\mu_k$ is an $n$-th root of unity, and multiplying any $\mu_k$ by the inverse of their product preserves all $n$-th powers and enforces $\prod_k\mu_k=1$.
\end{proof}

Next, we note that commutation of tensor product operators as in Eq.~(\ref{eq: tensor product operators}) implies commutation up to phase for the respective local unitaries in every tensor factor.

\begin{lemma}\label{lm: local phase commutation}
    With Lm.~\ref{lm: p-torsion gauge}, if $X_iX_jX^{-1}_iX^{-1}_j=\one$, then there exist $c_{ij}^{(k)}\in\zz_n$ for all $k\in[K]$ such that
    \begin{align*}
        U_{i,k}U_{j,k}
        =\zeta^{c_{ij}^{(k)}}_nU_{j,k}U_{i,k}\; ,\qquad
        \sum_kc_{ij}^{(k)}
        =0\; .
    \end{align*}
\end{lemma}

\begin{proof}
    We have $\bigotimes_{k=1}^KU_{i,k}U_{j,k}=\bigotimes_{k=1}^KU_{j,k}U_{i,k}$ which implies $U_{i,k}U_{j,k}=\lambda_kU_{j,k}U_{i,k}$ for scalars $\lambda_k$ with $\prod_{k=1}^K\lambda_k=1$. Repeatedly using that $U_{i,k}^n=\one$ by Lm.~\ref{lm: p-torsion gauge} gives $\lambda^n_k=1$ for all $k\in[K]$.
\end{proof}

The previous two lemmata hold for any integer $n\in\mathbb{N}$, the following lemma is special to $n$ odd.

\begin{lemma}\label{lm: trivial word power}
    Let $2\leq n\in\mathbb{N}$ be odd, $U^n_j=\one$ and $U_iU_j=\zeta^{c_{ij}}_nU_jU_i$ for $c_{ij}\in\zz_n$. Then $(U_1\cdots U_l)^n=\one$.
\end{lemma}

\begin{proof}
    Reordering the $n$ copies of $U_1\cdots U_l$ yields a phase with exponent proportional to $\binom{n}{2}$. Since $\binom{n}{2}=\frac{n(n-1)}2\equiv0\pmod n$, that phase is trivial, and the result follows with $U^n_j=\one$.
\end{proof}

For every row $r\in[M]$ in the LCS and with respect to a fixed variable order, define the row word
\begin{align*}
    W_r^{(k)}
    =\prod_{i=1}^NU_{i,k}^{A_{ri}}\; .
\end{align*}

\begin{theorem}\label{thm: trivial localisation}
    Let $2\leq n\in\mathbb{N}$ be odd and $\{X_i\}_{i=1}^N$ be a finite-dimensional quantum solution of the LCS $Ax=b\pmod n$ as in Eq.~(\ref{eq: tensor product operators}). Then there exist unique vectors $\beta^{(k)}\in\zz^M_n$ such that
    \begin{align}\label{eq: local beta}
        W_r^{(k)}
        &=\zeta^{\beta_r^{(k)}}_n\one\quad\forall r\in[M]&
        b&=\sum_k\beta^{(k)}
    \end{align}
    (in the local $n$-torsion gauge of Lm.~\ref{lm: p-torsion gauge}). Moreover, the classes $[\beta^{(k)}]\in\coker(A):=\zz^M_n/\im(A)$ are invariant under rephasing $U_{i,k}\mapsto\zeta^{s_i^{(k)}}_nU_{i,k}$ with $\sum_ks_i^{(k)}\equiv 0\pmod n$, and obey the relation
    \begin{equation}\label{eq: class sum}
        [b]=\sum_k[\beta^{(k)}]\; .
    \end{equation}
    In particular, if the LCS is classically unsatisfiable, then at least one tensor factor carrying a nontrivial obstruction class
    satisfies $\gcd(d_k,n)>1$ for $d_k=\dim(\cH_k)$.
\end{theorem}

\begin{proof}
    From $\bigotimes_kW_r^{(k)}=\zeta^{b_r}_n\one$ it follows that every unitary $W_r^{(k)}$ is a scalar multiple of the identity on the $k$-th tensor factor. The unitaries in one row phase-commute by Lm.~\ref{lm: local phase commutation}, hence, $(W_r^{(k)})^n=\one$ by Lm.~\ref{lm: trivial word power}, proving the first relation in Eq.~\eqref{eq: local beta}; the second follows by tensoring the local relations.
    
    Moreover, a normalised rephasing,  $U_{i,k}\mapsto\zeta^{s_i^{(k)}}_nU_{i,k}$ changes $\beta^{(k)}$ by $As^{(k)}$, hence, leaves its cokernel class invariant, and Eq.~\eqref{eq: class sum} thus follows from Eq.~\eqref{eq: local beta}.

    Finally, since $U_{i,k}^n=\one$, we have $\det(U_{i,k})=\zeta_n^{s_i^{(k)}}$ with $s_i^{(k)}\in\zz_n$ for all $i\in[N]$ and $k\in[K]$. Taking determinants in Eq.~(\ref{eq: local beta}) gives $A s^{(k)}=d_k\beta^{(k)}$ in $\zz_n^M$ and thus $d_k[\beta^{(k)}]=0$ in $\mathrm{coker}(A)$. If $\gcd(d_k,n)=1$, multiplication by $d_k$ is invertible over $\zz_n$, hence, $[\beta^{(k)}]=0$. On the other hand, classical unsatisfiability means $[b]\neq 0$, hence, the result follows with Eq.~(\ref{eq: class sum}).
\end{proof}

Thm.~\ref{thm: trivial localisation} does not imply that odd-modulus solutions cannot be written as simple tensors: any solution may be tensored with identities, and for odd $n$ one may also distribute the right-hand side among nontrivial tensor copies. The point is more specific. A classical obstruction cannot be created purely by cancellation between tensor factors which are individually trivial in $\coker(A)$.

Lm.~\ref{lm: trivial word power} fails at $n=2$ because $\binom{2}{2}=1$. If $U^2=V^2=I$ and $UV=-VU$, then $(UV)^2=-\one$. Thus a local row word formed from involutions can carry a fourth-root phase. In the Mermin--Peres square the two local qubits contribute phases $\pm i$, whose products give the global signs $\pm1$. There is therefore no factorised $\mathbb F_2$-valued decomposition analogous to Eq.~\eqref{eq: class sum}. This is the elementary parity mechanism underlying the exceptional binary local construction (see also App.~\ref{app: Pauli comparison}).

The conclusion should not be confused with a classification of binary operator solutions. Zhang, Pan and Liu construct finite-dimensional binary LCS strategies outside the Pauli and Clifford classes, their explicit solutions use genuinely multiqubit, non-tensorial involutions \cite{ZhangPanLiu}. In the opposite direction, Abramsky, Cercelescu and Constantin prove that completed scalar commutation groups admit faithful scalar-preserving Pauli representations, but their result does not identify an arbitrary LCS solution group with such a completed commutation group \cite{AbramskyCercelescuConstantin}. Thus these works delimit, rather than close, the binary reduction problem. Thm.~\ref{thm: trivial localisation} instead isolates what changes for factorised odd-modulus solutions, independent of any restriction on local operators.

\section{Quantum solutions from ray-context arrangements}\label{sec: incidence reduction}

Thm.~\ref{thm: trivial localisation} suggests to construct quantum solutions in a different way. We take inspiration from state-independent contextual arrangements, specifically arrangements of projectors on $\C^d$.

\begin{definition}\label{def: canonical LCS}
    Let $\{v\}_{i=1}^N$ be a set of vectors in $\C^d$, let $P_v=[v]$ denote the corresponding rank-one projector and let $\cC$ be a subset of all (maximal) contexts, that is, sets $C=\{P_{v_1},\ldots,P_{v_d}\}$ such that $P_{v_i}P_{v_j}=\delta_{ij}P_{v_i}$ and $\sum_{v\in C}P_v=\one$. $(\cP,\cC)$ is called a \emph{ray-context arrangement (in dimension $d$)}.
    
    Moreover, the \emph{LCS over $\zz_n$ associated with $(\cP,\cC)$} is defined by
    \begin{equation}\label{eq: canonical LCS}
        \sum_{v\in C}x_v\equiv 1\pmod n\; ,
        \qquad\forall C\in\cC\; .
    \end{equation}
\end{definition}

Throughout, all ray-context arrangements are finite unless stated otherwise.

\begin{proposition}\label{prop: canonical qsol}
    For every integer $n\in\mathbb{N}$ and ray-context arrangements $(\cP,\cC)$, the unitary operators
    \begin{equation}\label{eq: quantum solution to canonical LCS}
        X_v=\zeta^{P_v}_n
        =(\one-P_v)+\zeta_nP_v
        =\one+(\zeta_n-1)P_v
    \end{equation}
    form a $d$-dimensional quantum solution of the LCS associated with $(\cP,\cC)$ in Eq.~\eqref{eq: canonical LCS}.
\end{proposition}

\begin{proof}
    The spectral decomposition in Eq.~\eqref{eq: quantum solution to canonical LCS} gives $X^n_v=(\one-P_v)+\zeta^n_nP_v=\one$, and since within a context the projectors are mutually orthogonal,
    \begin{equation*}
        \prod_{v\in C}X_v
        =\prod_{v\in C}\zeta^{P_v}_n
        =\zeta^{\sum_{v\in C}P_v}_n
        =\zeta_n\one\; .\qedhere
    \end{equation*}
\end{proof}

Notably, there are various possible generalisations of Def.~\ref{def: canonical LCS} and Prop.~\ref{prop: canonical qsol}. In particular, the construction extends verbatim to projectors of arbitrary rank. (We will use this more general projector formulation only in the binary case in App.~\ref{app: Pauli comparison}.) More generally, one may consider non-maximal contexts and arbitrary unitaries. To settle the general existence problem of LCS with a quantum-classical gap, it will be sufficient to restrict to ray-context arrangements.

By construction, quantum solvability for the LCS associated to a set of rays and (a subset of) their contexts is automatic. This shifts the focus entirely to the problem of finding ray-context arrangements for which the associated LCS is classically unsatisfiable.

\subsection{General incidence obstructions}

In this and the next section we derive the condition for a LCS not to admit a classical solution. Our arguments apply to arbitrary LCS, independent of whether they admit a  quantum solution (specifically in the form of a ray-context arrangement). Still, for concreteness we will at times specialise to LCS associated with (and labelled by) ray-context arrangements.

A LCS is determined by its incidence data. We encode it via an incidence complex with
\begin{align*}
    \cC_1=\zz_n[\cC]\; ,\qquad
    \cC_0=\zz_n[\cX]\; ,
\end{align*}
with elementary chains denoted $[C]$ and $[X]$, respectively. In the case of ray-context arrangement, $\cX=\cP\subset\cP_1(\C^d)$ and thus $\cC_0=\zz_n[\cP]$. If $A$ is the operator-context incidence matrix, define
\begin{align*}
    \partial=A^T:\cC_1\to\cC_0\; .
\end{align*}
The algebraic dual map is $\delta=\partial^*=A:C_0^*\to C_1^*$. We regard $b$ as the $1$-cochain $\beta_b(\lambda)=\langle b,\lambda\rangle=b^T\lambda$ for $\lambda\in\cC_1$ and $\langle\cdot,\cdot\rangle$ the natural pairing between $C^*_1$ and $C_1$ (and similarly for $C^*_0$ and $C_0$).

\begin{proposition}\label{prop: cycle criterion}
    For an LCS $Ax\equiv b\pmod n$, the following are equivalent:
    \begin{enumerate}[label=(\roman*)]
        \item $Ax\equiv b\pmod n$ has no classical solution, equivalently $0\neq [b]\in\mathrm{coker}(A)=\zz^M_n/\mathrm{im}(A)$;
        \item there is $\lambda\in\ker(A^T)$ with $\beta_b(\lambda)=b^T\lambda\not\equiv0\pmod n$, equivalently $[\beta_b]\neq0$ in $\cC_1^*/\im(\delta)$.
    \end{enumerate}
\end{proposition}

\begin{proof}
    The statement follows by duality. More precisely, Pontryagin duality identifies $\zz^M_n$ with its character group $\widehat{\zz^M_n}$ via $y\mapsto\chi_y$ with $\chi_y(\lambda):=\zeta^{y^T\lambda}_n$, and this pairing is nondegenerate. For a subgroup $L\subset\zz^M_n$, let $L^\perp:=\{y\in\zz^M_n\mid y^T\lambda=0\; ,\forall\lambda\in L\}$. Then $\im(\delta)=(\ker(\partial))^\perp$.

    Now, under the standard basis pairing $y\mapsto\beta_y$, the matrix of the map $\delta=\partial^*$ is $A$, hence, $b\in\im(A)$ if and only if $\beta_b\in\im(\delta)$. Moreover, we have $\im(\delta)=(\ker(\partial))^\perp$ from above, hence, $\beta_b\notin\im(\delta)$ if and only if it pairs nontrivially with some $\lambda\in\ker(\partial)$.
\end{proof}

We call $\lambda\in\cC_1$ satisfying the condition in Prop.~\ref{prop: cycle criterion} an \emph{incidence cycle}. It certifies the obstruction,
\begin{align}\label{eq: general obstruction}
    0
    =x^T\left(A^T\lambda\right)
    =\left(Ax\right)^T\lambda
    =b^T\lambda
    \neq 0\; ,
\end{align}
more formally, viewing a classical solution $x$ as a $0$-cochain with coboundary $b=\delta x$, Eq.~(\ref{eq: general obstruction}) reads
\begin{align*}
    0
    =\langle x,\partial\lambda\rangle
    =\langle\delta x,\lambda\rangle
    =\langle b,\lambda\rangle
    \neq 0\; .
\end{align*}
For the ray-context LCS in Def.~\ref{def: canonical LCS}, we have $b=\mathbf{1}_\cC=(1,\cdots,1)^T$, and the pairing is the \emph{augmentation}
\begin{align}\label{eq: augmentation}
    \varepsilon(\lambda)
    \equiv\sum_{C\in\cC}\lambda_C\pmod n\; .
\end{align}
Moreover, for $\lambda=\mathbf{1}_\cC$ the condition in Prop.~\ref{prop: cycle criterion} becomes $A^T\mathbf{1}_\cC=0$ and $b^T\mathbf{1}_\cC\neq 0$. If every variable occurs in the same number $\kappa$ of contexts, then $A^T\mathbf{1}_\cC=\kappa\mathbf{1}_\cX$, and the first equation reduces to $n\mid \kappa$; if also $b=\mathbf{1}_\cC$, then the second reduces to $n\nmid|\cC|$. The obstructions of the $18$ ray-$9$ context refinement of the Mermin-Peres square as well as Mermin's pentagram are of this form. We will revisit those and compare our analysis with Pauli-based constructions to LCS (associated to the incidence structure of their symplectic polar space) more generally in App.~\ref{app: Pauli comparison}. Importantly, Prop.~\ref{prop: cycle criterion} and Eq.~(\ref{eq: general obstruction}) obtain much more generally than these examples make it seem. In fact, as we demonstrate next, even in the presence of symmetry the special cases above remain exceptional.

\subsection{Symmetric arrangements and their classical obstructions}\label{sec: symmetric arrangements}

In order to construct quantum solutions to inconsistent LCS, we will heavily rely on symmetry, in particular, we will construct ray-context arrangements that come with a finite group acting on its rays and contexts. More generally, the analysis in this section applies to any quantum solution $(\cX,\cC)$ with LCS associated to its underlying incidence relations (defined in the obvious way).

An \emph{incidence symmetry} is any permutation preserving the relation $X\in C$, that is, any symmetry of $A$. A \emph{geometric symmetry} is a projective unitary transformation preserving a given quantum solution $(\cX,\cC)$; every geometric symmetry induces an incidence symmetry, but the converse is generally not true. For ray-context arrangements with $X_v=\zeta^{P_v}_n$, geometric symmetries act non-projectively, $UP_vU^\dagger=P_{Uv}$ and thus $UX_vU^\dagger=X_{Uv}$, hence, no additional phase data is required. This is not true for general quantum solutions $(\cX,\cC)$, e.g. for $n$-qubit Pauli operators (see App.~\ref{app: Pauli comparison}).

Let $(\cX,\cC)$ be a quantum solution to some LCS. It is useful to distinguish several different notions of symmetry. First, $(\cX,\cC)$ is called \emph{operator-regular} if every unitary belongs to the same number $\kappa$ of contexts. Similarly, $(\cX,\cC)$ is called \emph{context-regular} if every context is incident in the same number of unitaries.\footnote{One may also consider flag-regular arrangements, where a flag is an incident pair $(v,C)$ with $v\in C$ (see Sec.~\ref{sec: HJS}).} Note that ray-context arrangements are always context-regular. In the above operator- and context-chain notation, regularity says $\partial\mathbf{1}_\cC=\kappa\mathbf{1}_\cX$. In particular, for ray-context arrangements, if $n\mid\kappa$ and $n\nmid |\cC|$, then $\mathbf{1}_{\cC}$ is an incidence cycle of nonzero augmentation, and the associated LCS over $\zz_n$ is classically unsatisfiable by Prop.~\ref{prop: cycle criterion}.

While various examples of contextuality proofs are based on regular group actions on an underlying incidence scenario, e.g. Pauli labels within the corresponding symplectic polar space or the ray construction in Sec.~\ref{sec: HJS} below, regularity is not necessary for classical unsatisfiability.

Below we will consider ray-context arrangements with non-transitive symmetry actions. In this case, operators are projectors $X_v=\zeta^{P_v}_n$ and contexts decompose into $G$-orbits,
\begin{align}\label{eq: orbit decomposition}
    \cP=\bigsqcup_{i=1}^s\cP_i\; ,\qquad
    \cC=\bigsqcup_{j=1}^t\cC_j\; .
\end{align}
Symmetry implies an equitable incidence partition: for fixed $i,j$, the number of contexts in $\cC_j$ containing a given ray of $\cP_i$ is independent of the chosen unitary.

For representatives $P_i\in\cP_i$ and $C_j\in\cC_j$, one thus defines the orbit-reduced incidence matrices
\begin{align}\label{eq: N}
    N_{ij}
    &=|\{C\in\cC_j\mid P_i\in C\}|\; ,&
    M_{ji}
    &=|\{X\in\cP_i\mid X\in C_j\}|=|C_j\cap\cP_i|\; .
\end{align}
Moreover, let $H_i=\Stab_G(P_i)$ and $K_j=\Stab_G(C_j)$ denote the respective stabiliser subgroups. Summing over all incident pairs $(P_i,C_j)\in\cP_i\times \cC_j$, we find that for every $i\in[s]$ and $j\in[t]$,
\begin{equation}\label{eq:doublecount}
    |\cP_i|N_{ij}=|\cC_j|M_{ji}
    \qquad\Longleftrightarrow\qquad
    N_{ij}=\frac{|H_i|}{|K_j|}M_{ji}\; ,
\end{equation}
where the latter follows since $|\cP_i|=\frac{|G|}{|H_i|}$ and $|\cC_j|=\frac{|G|}{|K_j|}$ by the orbit-stabiliser theorem.

With the above, we thus arrive at the following criterion for classical unsatisfiability.

\begin{proposition}\label{prop: orbit cycle}
    Let $(\cP,\cC)$ be a ray-context arrangement in $\C^d$, invariant under the action of (the unitary representation of) a finite group $G$. If there exists $\alpha\in\zz^t_n$ with $t$ as in Eq.~(\ref{eq: orbit decomposition}), such that
    \begin{align}\label{eq:orbitcriterion}
        N\alpha&\equiv0\pmod n&
        \sum_j|\cC_j|\alpha_j&\not\equiv0\pmod n\; ,
    \end{align}
    for $N$ as in Eq.~(\ref{eq: N}), then the associated LCS in Def.~\ref{def: canonical LCS} is classically unsatisfiable.
\end{proposition}

\begin{proof}
    For each context orbit $\cC_j$, define $u_j=\sum_{C\in \cC_j}[C]\in\zz_n[\cC]$ and $\lambda=\sum_j\alpha_j u_j$. Then $A^T\lambda=N\alpha\equiv0\pmod n$ as well as $\epsilon(\lambda)=\sum_j|\cC_j|\alpha_j\neq 0\mod n$. The result thus follows by Prop.~\ref{prop: cycle criterion}.
\end{proof}

In fact, Prop.~\ref{prop: orbit cycle} holds for any \emph{equitable} ray-context arrangement, that is, one with a partition of rays and contexts for which $N_{ij}$ and $M_{ji}$ are independent of the choice of $P_i$ and $C_j$, respectively. A group-orbit decomposition as in Eq.~(\ref{eq: orbit decomposition}) is always equitable, but an equitable partition need not arise from a group action.

Note that unlike Prop.~\ref{prop: cycle criterion}, Prop.~\ref{prop: orbit cycle} provides only a sufficient criterion for unsatisfiability in general.

\begin{proposition}\label{prop: symmetric solutions}
    Let $K=\ker(A)$ for the incidence map $A:\zz_n[\cP]\to\zz_n[\cC]$. If $H^1(G,K)=0$,\footnote{Recall that for a $G$-module $K$, a $1$-cocycle is a map $\eta:G\ra K$ with $\eta(gh)=\eta(g)+g\eta(h)$, while a $1$-coboundary has the form $\eta(g)=gk-k$. The first cohomology group is the quotient of the group of all cocycles by all coboundaries, $H^1(G,K)=Z^1(G,K)/B^1(G,K)$. It measures the obstruction to modifying a solution by an element of the homogeneous solution space $K=\ker(A)$ so as to make it $G$-invariant.} then the canonical LCS $Ax=\mathbf{1}\pmod n$ is classically satisfiable if and only if the orbit system $Mc=\mathbf{1}\pmod n$ is classically satisfiable. In particular, this holds whenever $\gcd(n,|G|)=1$.
\end{proposition}

\begin{proof}
    If $x_0$ is a classical solution of $Ax=\mathbf{1}\pmod n$, then $\eta(g)=gx_0-x_0\in K$ since $A(gx_0-x_0)=gAx_0-Ax_0=g\mathbf{1}-\mathbf{1}=0$, and $\eta(g)$ is a group $1$-cocycle, as $\eta(gh)=ghx_0-x_0=g(hx_0-x_0)+(gx_0-x_0)=g\eta(h)+\eta(g)$. Since $H^1(G,K)=0$, $\eta(g)=gk-k$ for some $k\in K$, hence,
    \begin{align*}
        g(x_0-k)
        =gx_0-gk
        =x_0+\eta(g)-gk
        =x_0+gk-k-gk
        =x_0-k\; ,
    \end{align*}
    that is, $x_0-k$ is $G$-invariant and therefore constant on ray orbits. Consequently, its orbit values define a solution to $Mc=\mathbf{1}\pmod n$. The converse is immediate.
    
    If $\gcd(n,|G|)=1$, then $|G|$ is invertible in $\zz_n$, hence, any solution $x$ can be averaged to obtain a $G$-invariant solution $\bar{x}=\frac{1}{|G|}\sum_{g\in G}gx$. As $\bar{x}$ is constant on every ray orbit $\cP_i$, let $c_i$ denote its value on $\cP_i$.  Evaluating $A\bar{x}=\mathbf{1}\pmod n$ on a context representative $C_j$ gives $Mc=\mathbf{1}\pmod n$.
\end{proof}

\subsection{$\zz_n$-valued frame functions}

We will also exploit the following connection with frame functions.

\begin{definition}\label{def: frame function}
    A $\zz_n$-valued frame function of weight $w\in\zz_n$ on $\C^d$ is a map $f:\cP_1(\C^d)\to\zz_n$ such that $\sum_{P\in C}f(P)=w$ for every maximal orthonormal set of rank-one projectors.

    $f$ is called \emph{constant} if $f(P)=c$ for every rank-one projector $P\in\cP_1(\C^d)$ and $c\in\zz_n$.
\end{definition}

A $\zz_n$-valued frame function is constant if and only if its additive extension to finite-rank projectors is proportional to rank, that is, $f(P)=c\cdot\mathrm{rank}(P)$. In particular, its weight is $w=dc$.

\begin{lemma}\label{lm: ff + invertible weight -> csol}
    Let $f$ be a $\zz_n$-valued frame function of weight $w$ on $\C^d$, and let $Ax\equiv b\pmod n$ be a LCS with a $d$-dimensional quantum solution. Then there is $x_f\in\zz^N_n$ such that $Ax_f=wb$. In particular, every frame function with invertible weight yields a classical solution to every such LCS.
\end{lemma}

\begin{proof}
    We consider the additive extension of $f$ to finite-rank projectors. For an $n$-torsion unitary $X=\sum_{a\in\zz_n}\zeta^a_nQ_a$ (of a quantum solution $\{X_i\}_i$) define $x_f(X)=\sum_{a\in\zz_n}af(Q_a)$. Moreover, for every constraint $\prod_{i=1}^NX_i^{A_{ri}}=\zeta^{b_r}_n\one$, choose a common rank-one spectral resolution $\{P_k\}$. If $a_{ik}$ is the eigenvalue exponent of $X_i$ on $P_k$, then $\sum_{i=1}^NA_{ri}a_{ik}\equiv b_r\pmod n$ for every $k$, hence,
    \begin{align*}
        \sum_{i=1}^NA_{ri}x_f(X_i)
        =\sum_{k=1}^df(P_k)\sum_{i=1}^NA_{ri}a_{ik}
        =b_r\sum_{k=1}^df(P_k)\equiv wb_r\pmod n\; .
    \end{align*}
    Clearly, if $w$ is invertible, $w^{-1}x_f$ defines a classical solution to the LCS.
\end{proof}

In particular, the existence of a $\zz_n$-valued frame function on $\cP_1(\C^d)$ of weight $1$ provides a classical solution to the LCS associated to any ray-context arrangement in $\C^d$. Conversely:

\begin{lemma}\label{lm: qsol from frame function}
    Let $2\leq n\in\mathbb{N}$ and $d\in\mathbb{N}$. If there exists no $\zz_n$-valued frame function of weight one on $\C^d$, then there is a finite family of contexts $\cC$ (complete orthogonal rank-one resolutions of $\C^d$) such that the LCS $\sum_{P\in C}x_P\equiv 1\pmod n$ for all $C\in\cC$ (see Def.~\ref{def: canonical LCS}) is classically unsatisfiable.
\end{lemma}

\begin{proof}
    Let $\PHone$ denote a set of rank-one projectors on $\cH=\C^d$ and let $\CH$ denote the set of all complete orthogonal contexts. Consider the product space $\Omega=\zz^\PHone_n$, where $\zz_n$ is given the discrete topology. Since $\zz_n$ is finite, it is compact, hence, $\Omega$ is compact by Tychonoff's theorem.
    
    For each $C\in\CH$, let $\Omega_C=\left\{x\in\Omega\mid\sum_{P\in C}x_P=1\right\}$. The set $\Omega_C$ is clopen: it is the inverse image of $\{1\}$ under a continuous map depending only on the finitely many coordinates indexed by the rays in $C$. Now, suppose for contradiction that every finite collection of the canonical context equations were satisfiable. Then for every finite $\cC\subseteq\CH$ one has $\bigcap_{C\in\cC}\Omega_C\neq\emptyset$, that is, the family $\{\Omega_C\mid C\in\CH\}$ has the finite intersection property, and compactness of $\Omega$ therefore implies
    \begin{align*}
        \bigcap_{C\in\CH}\Omega_C\neq\emptyset\; .
    \end{align*}
    Yet, an element of this intersection is a $\zz_n$-valued frame function of weight one, contrary to our assumption. Hence, there exists a finite subfamily $\cC\subseteq\CH$ whose associated LCS is inconsistent.
\end{proof}

\begin{theorem}\label{thm: frame functions vs qsols}
    Let $2\leq n\in\mathbb{N}$ and $d\in\mathbb{N}$, then the following are equivalent:
    \begin{enumerate}
        \item there exists a classically unsatisfiable LCS over $\zz_n$ with $d$-dimensional quantum solution;
        \item there exists no $\zz_n$-valued frame function on $\C^d$ of weight one;
        \item there exists a finite classically unsatisfiable canonical rank-one projector LCS in dimension $d$.
    \end{enumerate}
\end{theorem}

\begin{proof}
    $(1)\Rightarrow(2)$ follows from Lm.~\ref{lm: ff + invertible weight -> csol}, the implication $(2)\Rightarrow(3)$ is Lm.~\ref{lm: qsol from frame function}, and $(3)\Rightarrow(1)$ is obvious.
\end{proof}

Lm.~2 in Ref.~\cite{HardingJagerSmith2005}, (using the explicit ray--context arrangement in Lm.~\ref{lm: HJS construction} below), proves that $\zz_n$-valued frame functions on $\cP_1(\C^d)$ with $d=2n$ have weight zero. Thm.~\ref{thm: frame functions vs qsols} thus asserts the existence of classically unsatisfiable LCS over $\zz_n$ with quantum solutions in dimension $d=2n$. In fact, the construction in Ref.~\cite{HardingJagerSmith2005} provides an explicit such LCS. In the next section, we review and complement their construction to fully resolve the existence problem of such LCS for all $n,d$.

\section{Existence of quantum solutions to classically unsatisfiable LCS}\label{sec: qsols}

In this section, we provide explicit constructions for LCS over $\zz_n$ for $2\leq n\in\mathbb{N}$ with no classical solution but a $d$-dimensional quantum whenever $\gcd(n,d)>1$; we thereby solve the open existence problem of LCS with a quantum-classical gap completely and constructively. We first leverage the construction from Ref.~\cite{HardingJagerSmith2005} in Sec.~\ref{sec: HJS}, which yields a family of LCS over $\zz_n$ with quantum-classical gap in dimension $d=kn$ for $k\geq 2$. The remaining open case $d=p$ prime is covered in Sec.~\ref{sec: d=p prime}, and extended to arbitrary integers $2\leq n\in\mathbb{N}$ in Sec.~\ref{sec: full solution}. We also establish the existence of quantum solutions to classically unsatisfiable LCS over arbitrary finite fields in Sec.~\ref{sec: finite fields}.

\subsection{Harding--Jager--Smith construction}\label{sec: HJS}

Thm.~\ref{thm: frame functions vs qsols} together with the partial characterisation of finite group-valued measures on Hilbert lattices in Ref.~\cite{HardingJagerSmith2005} establishes the existence of LCS with a quantum-classical gap for all $n\geq 2$ and $d=kn$ with $k\geq 2$. To prove their result they construct an explicit family of ray-context arrangements.

\begin{lemma}[Lm.~1 in \cite{HardingJagerSmith2005}]\label{lm: HJS construction}
    For any $2\leq n\in\mathbb{N}$, there is a set $Z$ of vectors in $\R^{2n}$ such that
    \begin{enumerate}
        \item $Z$ has $2(n+1)^2$ vectors.
        \item $Z$ can be covered by a family of $(n+1)^2$ blocks $B_{ij}$ for $0\leq i,j\leq n$.
        \item Each vector in $Z$ occurs in exactly $n$ of the blocks $B_{ij}$ where $0\leq i,j\leq n$.
    \end{enumerate}
\end{lemma}

Here, `block' is synonymous with basis and our `context' in Def.~\ref{def: canonical LCS}.

\begin{theorem}\label{thm: HJS LCS}
    For every $2\leq n\in\mathbb{N}$ and every $d=kn$ with $k\geq 2$, there exists a LCS over $\zz_n$ with a $d$-dimensional quantum solution, yet no classical solution.
\end{theorem}

\begin{proof}
    For $k=2$, consider the LCS associated to the arrangement in Lm.~\ref{lm: HJS construction}, consisting of a variable $x_z$ for every ray $[z]$ with $z\in Z$. Assume that these variables satisfy the constraints
    \begin{align}\label{eq: HJS LCS}
        \sum_{z\in B_{ij}}x_z
        \equiv 1\pmod n\; 
    \end{align}
    for every $0\leq i,j\leq n$. Then summing over all constraints yields the contradiction
    \begin{align*}
        0\equiv
        n\sum_{z\in Z}x_z
        =\sum_{i,j=0}^n\sum_{z\in B_{ij}}x_z
        \equiv\sum_{i,j=0}^n 1
        =(n+1)^2
        \equiv 1\pmod n\; 
    \end{align*}
    where we used that every variable appears in $n$ bases by Lm.~\ref{lm: HJS construction}.~3, Eq.~(\ref{eq: HJS LCS}) and Lm.~\ref{lm: HJS construction}.~2.
    
    For $d=kn$ with $k>2$, Thm.~3 in Ref.~\cite{HardingJagerSmith2005} implies that every frame function is of the form $f(P)=c\cdot\rank(P)$ for $c\in\zz_n$ with weight $dc\equiv 0\pmod n$. The result thus follows with Thm.~\ref{thm: frame functions vs qsols}.
\end{proof}

Lm.~\ref{lm: HJS construction} yields an explicit construction of a LCS over $\zz_n$ with quantum-classical gap for all $d=kn$ with $k=2$, and Thm.~3 in Ref.~\cite{HardingJagerSmith2005} (non-constructively) implies the existence for $k>2$ by Thm.~\ref{thm: frame functions vs qsols}. For an explicit construction, we extend the above arrangement to all dimensions $d=kn$ with $n\geq 3$ and $k\geq 2$. (For an alternative construction of such LCS, see also App.~\ref{app: conference LCS}.)\\

\textbf{Quantum solutions to Harding--Jager--Smith (HJS) LCS.} For all $n\geq 2$ and $L\geq 3$,\footnote{The case $L=2$ is treated in Lm.~\ref{lm: HJS construction}.} let
\begin{align*}
    \cH
    =\bigoplus_{l=1}^L\cH_l\; ,
    \qquad \dim(\cH_l)=n\; ,
\end{align*}
and fix an orthonormal basis $B_l$ for each $\cH_l$. We include the global coordinate basis $B_0=B_1\cup\cdots\cup B_L$. Moreover, for every pair $l,m\in[L]$, $l<m$, place a copy of the configuration in Lm.~\ref{lm: HJS construction} in $\cH_l\oplus\cH_m$, consisting of $(n+1)^2$ orthogonal bases $B^{lm}_{ij}$, and $2(n+1)^2$ rays, each of which occurs in exactly $n$ of these bases. We complete every such $\cH_l\oplus\cH_m$-basis to a basis of $\cH$ by setting
\begin{align*}
    \widetilde B^{lm}_{ij}
    =B^{lm}_{ij}\cup\bigcup_{c\neq l,m}B_c\; .
\end{align*}
The resulting arrangement thus consists of $Ln+\binom{L}{2}2(n+1)^2$ rays and $1+\binom{L}{2}(n+1)^2$ bases in total.

We show that the associated LCS $\sum_{v\in B}x_v\equiv 1\pmod n$ for every basis $B$ is inconsistent. Let
\begin{align*}
    s_l=\sum_{v\in B_l}x_v\in\mathbb Z_n\; .
\end{align*}
Then $\sum_{l=1}^L s_l\equiv 1\pmod n$ for the coordinate context. Moreover, for a fixed pair $l<m$, sum the $(n+1)^2$ equations corresponding to the $\widetilde B^{lm}_{ij}$. Every HJS ray occurs exactly $n$ times, hence, cancels modulo $n$, whereas every coordinate ray in $B_c$, $c\neq l,m$, occurs $(n+1)^2\equiv1\pmod n$ times, hence,
\begin{align*}
    \sum_{c\neq l,m}s_c
    \equiv 1\pmod n\; .
\end{align*}
Comparing the two expressions, we find that $s_l+s_m\equiv 0\pmod n$ for all $l\neq m$.

For $n\geq 3$, combining these relations for three distinct indices, $s_{l_1}+s_{l_2}=s_{l_1}+s_{l_3}=s_{l_2}+s_{l_3}\equiv 0\pmod n$, yields $2s=2s_l\equiv 0\pmod n$ for all $l\in[L]$, $Ls\equiv 1\pmod n$ and thus $0\equiv 2\pmod n$. Consequently, for $n\geq 3$ and $L\geq 3$ the LCS is classically unsatisfiable: $2s\equiv 0\pmod n$ implies $s=0$ for $n$ odd, and $s=0$ or $s=\frac{n}{2}\neq 1$ for $n>2$ even, hence, $0=2Ls\equiv 2\pmod n$.

For $n=2$, $2s\equiv 0\pmod 2$ is trivially satisfied, hence, the LCS has a solution only if $Ls\equiv 1\pmod 2$ which is satisfied for $s=1$ if and only if $L$ is odd. In fact, a solution to this LCS is obtained by setting $x_v=0$ for all HJS rays and on coordinate rays such that $\sum_{v\in B_l}x_v\equiv s\pmod n$.\\

\textbf{Symmetry of the HJS ray-context arrangement.} The ray-context arrangement in Lm.~\ref{lm: HJS construction} is highly symmetric. For $n\geq 2$, Ref.~\cite{HardingJagerSmith2005} consider the $2n$-dimensional subspace
\begin{align*}
    W=\{e,f\}^{\perp}\subset\R^{2n+2},\qquad
    e=\sum_{i=0}^{n}e_i,\qquad
    f=\sum_{i=0}^{n}f_i,
\end{align*}
where $e_0,\ldots,e_n,f_0,\ldots,f_n$ is the standard orthonormal basis of $\R^{2n+2}$, and define
\begin{align*}
    s_i=e-(n+1)e_i\; ,\qquad
    t_i=f-(n+1)f_i\; ,
\end{align*}
for every $0\leq i\leq n$. The sets $\{s_i\}_{i=0}^{n}$ and $\{t_i\}_{i=0}^{n}$ thus constitute two mutually orthogonal regular $n$-simplices in $W$. The set of rays in Lm.~\ref{lm: HJS construction} is given by
\begin{align*}
    Z=\{u_{ij},v_{ij}\mid 0\leq i,j\leq n\}\; ,
\end{align*}
where $u_{ij}=\rho s_i+t_j$ and $v_{ij}=s_i-\rho t_j$ with $\rho=n^{-1/2}$, and for $0\leq i,j\leq n$ the set of contexts by
\begin{align*}
    B_{ij}
    =\{u_{ik}\mid k\neq j\}\,\sqcup\,\{v_{kj}\mid k\neq i\}\; .
\end{align*}
It follows that the group $S_{n+1}\times S_{n+1}$ acts by
\begin{align*}
    (\sigma,\tau)[u_{ij}]
    =[u_{\sigma(i),\tau(j)}]\; ,\qquad
    (\sigma,\tau)[v_{ij}]
    =[v_{\sigma(i),\tau(j)}]\; ,\qquad
    (\sigma,\tau)B_{ij}
    =B_{\sigma(i),\tau(j)}\; ,
\end{align*}
hence, the projective group $G_{\mathrm{HJS}}\cong(S_{n+1}\times S_{n+1})\rtimes C_2$ is transitive on rays and on contexts. In fact, it is also flag-transitive: $S_{n+1}\times S_{n+1}$ is transitive on flags of a fixed $u$- or $v$-type, while the orthogonal map $J(s+t)=K^{-1}t-Ks$ with $Ks_i=t_i$ linear exchanges the two types.\footnote{Notably, $J$ defines a complex structure, in particular, $J^2=-\one$.}

It is interesting to ask whether a similar construction exists for $d=p$, that is, whether one can construct a ray-context arrangement in $\C^d$ which is ray- and context-transitive (or even flag-transitive), whose associated LCS over $\zz_p$ has a quantum-classical gap. We will not answer this question here. Yet, as outlined in Sec.~\ref{sec: symmetric arrangements}, this is not necessary to solve our present problem of deciding whether LCS with a quantum-classical gap exist. Indeed, relaxing the transitivity conditions will allow us to cover the remaining case $d=p$.

\subsection{The case $d=p$ prime}\label{sec: d=p prime}

From the perspective of solving the existence problem for LCS with quantum but no classical solution, the construction in Sec.~\ref{sec: HJS} essentially only leaves open the case $d=p$. We cover this case in this section and give the resulting full solution for rings $\zz_n$ in Sec.~\ref{sec: full solution} (and for finite fields in Sec.~\ref{sec: finite fields}). Throughout, we identify the cyclic ring $\zz_p$ with the finite field $\Fp$ for prime $p$, and use the latter notation when emphasising linear-algebraic structure.

\subsubsection{Quantum solutions in prime dimension $p\geq 5$}\label{sec: p prime}

We first give a construction for every odd prime $p\geq 5$. The construction is real, although the associated canonical quantum solution in Prop.~\ref{prop: canonical qsol} is naturally complex.\\

\textbf{Notation.} Denote by
\begin{align}\label{eq: Gp}
    G_p
    =\overline{W(B_p)}
    =W(B_p)/\{\pm 1\}
    \cong (\mathbb Z_2)^{p-1}\rtimes S_p
\end{align}
the projective hyperoctahedral group acting on $\RP^{p-1}$ (by coordinate permutations and relative inversions), of cardinality $|G_p|=2^{p-1}p!$. For $\cC_1=\Fp[\cC]$ and $\cC_0=\Fp[\cP]$, we define
\begin{align*}
    \partial[C]
    :=\sum_{v\in C}[v]\; ,\qquad
    \epsilon\left(\sum_{C\in\cC}\lambda_C[C]\right)
    :=\sum_{C\in\cC}\lambda_C\; ,
\end{align*}
as before. For a ray $[x]\in\RP^{p-1}$ define its orbit sum, or unnormalised $G_p$-twirl, by
\begin{align*}
    \cT_{G_p}(x)
    =\sum_{g\in G_p}[gx]
    =|\Stab_{G_p}([x])|\sum_{y\in G_p[x]}[y]\in\cC_0\; ;
\end{align*}
similarly, the unnormalised $G_p$-twirl on context chains $\sigma\in\cC_1$ is defined by
\begin{align*}
    \cT_{G_p}(\sigma)
    =\sum_{g\in G_p}g\sigma\in\cC_1\; .
\end{align*}

\textbf{Transport lemma}. The key geometric observation is that, in odd dimension, every real ray with at least one zero coordinate (equivalently, with proper support) can be transported to the coordinate orbit by finite context chains of zero augmentation.

\begin{lemma}\label{lm: support transport}
    Let $p$ be odd, $G_p$ the projective hyperoctahedral group in Eq.~(\ref{eq: Gp}), and let $0\neq x\in\R^p$ have at least one zero coordinate. Then there exists a finite context chain $\eta_x\in\cC_1$ such that
    \begin{align}\label{eq: support transport}
        \partial\eta_x=\cT_{G_p}(x)-\cT_{G_p}(e_1)\; ,
        \qquad
        \epsilon(\eta_x)=0\; .
    \end{align}
\end{lemma}

\begin{proof}
    Suppose that $x_l=0$. Since $p-1$ is even, pair the coordinates different from $l$. On each pair choose the signed transformation,
    \begin{align*}
        J_\varepsilon=
        \begin{pmatrix}
          0&\varepsilon\\
          -\varepsilon&0
        \end{pmatrix},
        \qquad
        \varepsilon\in\{\pm1\},
    \end{align*}
    and let $J$ be their direct sum, fixing the $l$-th coordinate. $J$ is a signed permutation, in particular, $J\in G_p$, and $J^2=-\one$ on the hyperplane $e_l^\perp$, hence, $x\perp Jx$.
    
    For any $\theta\in[0,2\pi)$ set $u=\cos(\theta)x+\sin(\theta)Jx$. Since $u_l=0$, $Ju=-\sin(\theta)x+\cos(\theta)Jx$, hence,
    \begin{align*}
        u\perp Ju\; ,\qquad
        \mathrm{span}\{x,Jx\}
        =\mathrm{span}\{u,Ju\}\; .
    \end{align*}
    Choose a common orthogonal complement $B$ and form the two orthogonal contexts
    \begin{align*}
        C=\{[x],[Jx]\}\cup B\; ,\qquad
        C'=\{[u],[Ju]\}\cup B\; .
    \end{align*}
    Since $J\in G_p$, we have
    \begin{align*}
        \partial\cT_{G_p}(C-C')
        =\cT_{G_p}(x)+\cT_{G_p}(Jx)-\cT_{G_p}(u)-\cT_{G_p}(Ju)
        =2\cT_{G_p}(x)-2\cT_{G_p}(u)\; ,
    \end{align*}
    Moreover, since $p$ is odd, $\eta(x,J,\theta):=\frac{1}{2}\cT_{G_p}(C-C')$ is well-defined over $\Fp$ and satisfies
    \begin{align}\label{eq: basic support transport}
        \partial\eta(x,J,\theta)=\cT_{G_p}(x)-\cT_{G_p}(u)\; ,\qquad
        \epsilon(\eta(x,J,\theta))=0\; .
    \end{align}
    
    It remains to show that the above rotations generate arbitrary coordinate-plane rotations while preserving the zero coordinate $x_l$. Fix $i,j\neq l$, and choose a pairing of the coordinates different from $l$ containing the pair $\{i,j\}$. Let $J_+$ and $J_-$ have the same sign $\epsilon$ on $\{i,j\}$ but opposite signs on every other pair. Writing $R_J(\theta)=\cos(\theta)\one_{e^\perp_l}+\sin(\theta)J$, one obtains (on $e_l^\perp$)
    \begin{equation*}
        R_{J_-}(\theta)R_{J_+}(\theta)
        =R_{ij}(2\theta)\; ,
    \end{equation*}
    where $R_{ij}(2\theta)$ is the ordinary elementary rotation in the $(i,j)$-plane and is the identity on all other coordinates. Indeed, on the distinguished pair the two quarter-turns $J_\pm$ coincide, whereas on every other pair they are negatives of one another, hence, the latter rotations cancel and the resulting rotation in the $(i,j)$-plane has the sum of the angles. Applying Eq.~(\ref{eq: basic support transport}) successively therefore produces, for every rotation $R_{ij}(\phi)$, a finite zero-augmentation context chain with boundary
    \begin{align*}
        \cT_{G_p}(x)-\cT_{G_p}(R_{ij}(\phi)x)\; .
    \end{align*}
    After at most $p-2$ elementary rotations in $e_l^\perp$, $x=x^{(0)}$ is rotated into a scalar multiple of a coordinate vector $x^{(m)}=\|x\|e_r$. Let $\eta_s$ be the context chain satisfying $\partial\eta_s=\cT_{G_p}(x^{(s-1)})-\cT_{G_p}(x^{(s)})$ and $\epsilon(\eta_s)=0$, then the context chain $\eta=\sum_{s=1}^m\eta_s$ has boundary 
    \begin{align*}
        \partial\eta
        =\sum_{s=1}^m \cT_{G_p}(x^{(s-1)})-\cT_{G_p}(x^{(s)})
        =\cT_{G_p}(x)-\cT_{G_p}(x^{(m)})
        =\cT_{G_p}(x)-\cT_{G_p}(e_r)\; ,
    \end{align*}
    and since all coordinate vectors lie in the same $G_p$-orbit, this proves Eq.~(\ref{eq: support transport}).
\end{proof}

It follows from Lm.~\ref{lm: support transport} that unnormalised $G_p$-twirled incidence relations collapse to the coordinate orbit in the quotient $\cC_0/\partial(\mathrm{ker}(\epsilon))$. To obtain a nontrivial obstruction, we need to connect this collapsed proper-support class to an orbit with different stabiliser group.\\

\textbf{Diagonal orbit.} Let $e=(1,\ldots,1)^T$ and for $p\geq5$ define
\begin{align}\label{eq: crossing}
    t=\sqrt{\frac{p-2}{2}},
    \qquad
    y=(0,1,-1,t,-t,0,\ldots,0)^T\; .
\end{align}
It follows that $e\cdot y=0$, $\|y\|^2=2+2t^2=p=\|e\|^2$. Next, define
\begin{align*}
    u=e+y\; ,
    \qquad
    v=e-y\; .
\end{align*}
Clearly, $u\perp v$ and $\mathrm{span}\{e,y\}=\mathrm{span}\{u,v\}$. Moreover, $y_1=0$, $u_3=0$, $v_2=0$, hence, $y,u,v$ all have proper support. Choose a common orthogonal complement $B$ and form the contexts
\begin{align}\label{eq: diagonal crossing}
    C_*=\{[e],[y]\}\cup B\; ,
    \qquad
    C_*'=\{[u],[v]\}\cup B\; .
\end{align}
Setting $\delta:=\cT_{G_p}(C_*-C_*')$ we have $\partial\delta=\cT_{G_p}(e)+\cT_{G_p}(y)-\cT_{G_p}(u)-\cT_{G_p}(v)$ and $\epsilon(\delta)=0$. Now, note the projective stabiliser group of the diagonal ray $e$ has order $|\Stab_{G_p}([e])|=p!$, and thus
\begin{align*}
    \cT_{G_p}(e)
    =p!\sum_{w\in G_p[e]}[w]
    \equiv 0\pmod p\; .
\end{align*}
Moreover, by Lm.~\ref{lm: support transport}, there exist chains $\eta_y,\eta_u,\eta_v$ satisfying $\partial\eta_y=\cT_{G_p}(y)-\cT_{G_p}(e_1)$, $\partial\eta_u=\cT_{G_p}(u)-\cT_{G_p}(e_1)$, $\partial\eta_v=\cT_{G_p}(v)-\cT_{G_p}(e_1)$ and $\epsilon(\eta_y)=\epsilon(\eta_u)=\epsilon(\eta_v)=0$, hence,
\begin{align*}
    \partial\bigl(\delta-\eta_y+\eta_u+\eta_v\bigr)
    =-\cT_{G_p}(e_1)\; .
\end{align*}
Next, let $C_0=\{[e_1],\ldots,[e_p]\}$ be the coordinate context. Since the coordinate-ray orbit has size $p$, by the orbit-stabiliser theorem, we have $|\Stab_{G_p}([e_1])|=\frac{|G_p|}{p}=2^{p-1}(p-1)!$, and thus
\begin{align*}
    \cT_{G_p}(e_1)
    =2^{p-1}(p-1)!\sum_{i=1}^p[e_i]
    =2^{p-1}(p-1)!\,\partial C_0\; .
\end{align*}
Consequently, we obtain a finite incidence cycle (that is, $\partial\lambda=0$) of the form:
\begin{align}\label{eq: incidence cycle}
    \lambda
    =\delta-\eta_y+\eta_u+\eta_v+2^{p-1}(p-1)!\,C_0\; .
\end{align}
On the other hand, since $\epsilon(\delta)=\epsilon(\eta_y)=\epsilon(\eta_u)=\epsilon(\eta_v)=0$, Fermat's and Wilson's theorems give
\begin{align*}
    \epsilon(\lambda)
    =\epsilon(2^{p-1}(p-1)!\,C_0)
    =2^{p-1}(p-1)!
    \equiv-1\not\equiv0\pmod p\; .
\end{align*}

In summary, this proves the following.

\begin{theorem}\label{thm: d=p>3 prime}
    For every odd prime $p\geq5$, there exists a ray-context arrangement $(\cP,\cC)$ in $\R^p$ for which the associated LCS in Def.~\ref{def: canonical LCS} of the form $\sum_{P\in C}x_P\equiv 1\pmod p$ for all $C\in\cC$ is inconsistent, yet admits a $p$-dimensional quantum solution of the form in Prop.~\ref{prop: canonical qsol}.
\end{theorem}

\begin{proof}
    The chain $\lambda$ in Eq.~(\ref{eq: incidence cycle}) satisfies $\partial\lambda=0$ and $\epsilon(\lambda)\neq 0$, hence, the result follows from Prop.~\ref{prop: cycle criterion}. On the other hand, the unitary operators $X_P$ solve every context equation by Prop.~\ref{prop: canonical qsol}.
\end{proof}

In other words, the obstruction no longer follows from the cardinalities of single ray and context orbits alone, but from the orbit-incidence relations between several such orbits, with the relevant multiplicities controlled by their stabiliser orders.\\

\textbf{Remarks.} By Prop.~\ref{prop: determinant} the dimension $d=p$ is minimal for quantum solutions to LCS over $\zz_p$ that have no classical solution.

Note also that the case $p=3$ is excluded from the above construction, since Eq.~(\ref{eq: crossing}) from the diagonal ray $e=(1,\ldots,1)$ to proper-support rays uses five coordinates.

Before we consider the case $d=p=3$, and extend the result to arbitrary integers, we provide a bound on the size of the arrangement. Every elementary rotation requires two contexts and their orbits, hence, up to $2|G_p|$ contexts in total. The full rotation to a coordinate ray requires at most $2(p-2)$ elementary rotations, hence, up to $12(p-2)|G_p|$ contexts in total since we use rotations for $y,u,v$. Finally, we need to add another $2|G_p|$ contexts of the form in Eq.~(\ref{eq: diagonal crossing}), which together with the coordinate contexts results in the upper bound
\begin{align*}
    |\cC_p|\leq
    (12p-22)|G_p|+1
    =(12p-22)2^{p-1}p!+1\; .
\end{align*}
Since every context contains $p$ elements, this also yields a bound on the rays of the arrangement, $|\cP|\leq p(12p-22)2^{p-1}p!+p$. Of course, these bounds do not take into account stabilisers and overlaps, hence, the actual number of contexts may be considerably smaller.\\

\textbf{Symmetry.} The construction in Thm.~\ref{thm: d=p>3 prime} has a different symmetry structure to that of the HJS arrangement in Lm.~\ref{lm: HJS construction}: the ray-context arrangement in Thm.~\ref{thm: d=p>3 prime} decomposes not into a single, but several $G_p$-orbits. In particular, the coordinate ray $[e_1]$ and the diagonal ray $[e]=[(1,\ldots,1)^T]$ belong to different orbits, and have stabiliser subgroups of cardinality
\begin{align}\label{eq: stabiliser cardinalities}
    |(G_p)_{e_1}|
    &=2^{p-1}(p-1)!\not\equiv0\pmod p &
    |(G_p)_e|
    &=p!\equiv0\pmod p\; .
\end{align}
The final obstruction is therefore a weighted combination of several context orbits rather than a single all-context sum. More explicitly, the proof uses the different $p$-divisibility of the two stabilisers in Eq.~(\ref{eq: stabiliser cardinalities}). Recall that the unnormalised $G$-twirl is given by $\cT_G[v]=|G_v|\sum_{w\in G[v]}[w]$, hence, $\cT_{G_p}[e]=0$ over $\Fp$, whereas $\cT_{G_p}[e_1]\neq 0$. The construction transports the former orbit to the latter by the context chains in Lm.~\ref{lm: support transport}, and closes the remaining boundary with the coordinate context.

We point out one more subtlety. Since $p\mid |G_p|$, the \emph{normalised} $G_p$-twirl, $|G_p|^{-1}\sum_{g\in G_p}g$, is not defined over $\Fp$, that is, one cannot average a classical solution to obtain a $G_p$-invariant one. Only unnormalised orbit sums are well-defined, for which $\epsilon(\cT_{G_p}\sigma)=|G_p|\,\epsilon(\sigma)\equiv 0\pmod p$. This is why the orbit reduction criterion in Prop.~\ref{prop: orbit cycle} is automatically complete when $p\nmid|G|$, but need not be complete in the modular case. More generally, the obstruction to symmetrizing a classical solution is measured by the corresponding group-cohomology class in $H^1(G,\ker A)$ (see Prop.~\ref{prop: symmetric solutions}).

\subsubsection{LCS over $\mathbb{F}_p$ with quantum solution in $d=3$}\label{sec: d=p=3}

The construction in the last section cannot be applied to the case $p=3$. To close this gap, we give a small and highly symmetric ray--context arrangement in $\C^{3}$ whose associated LCS over $\zz_3$ is classically unsatisfiable, while admitting the canonical qutrit quantum solution of Prop.~\ref{prop: canonical qsol}. The construction is closely related to the Peres--Penrose family of Kochen--Specker configurations with $33$ rays and $16$ contexts \cite{Peres1991,Penrose1994} (see App.~\ref{app: d=p=3} for details).

Let $\alpha=\sqrt{-2}=i\sqrt{2}$ such that $\alpha^{2}=-2$, $\overline{\alpha}=-\alpha$, $\alpha\overline{\alpha}=2$, and consider vectors in $\C^{3}$ with coordinates in the imaginary quadratic ring $\zz[\sqrt{-2}]=\zz+i\sqrt{2}\zz$, that is, we consider vectors in the lattice $\zz[\sqrt{-2}]^3\subset\C^3$; we write $[v]$ for the ray generated by a nonzero vector $v\in\C^{3}$.

Moreover, let $W=W(B_{3})=\zz^3_2\rtimes S_3$ be the signed permutation group acting on $\C^{3}$ by permuting the three coordinates and changing their signs independently. The central element $-\one$ acts trivially on rays, so the action on projective rays factors through $W/{\pm\one}$. Consider the four $W$-orbits of rays
\begin{align}\label{eq: ray orbits}
\begin{split}
    \cP_1
    &=W\cdot[1,0,0]\; ,\hspace{1.55cm} |\cP_1|=3\; ,\\
    \cP_2&=W\cdot[1,1,1]\; ,\hspace{1.55cm} |\cP_2|=4\; ,\\
    \cP_3&=W\cdot[1,\alpha,0]\; ,\hspace{1.5cm} |\cP_3|=12\; ,\\
    \cP_4&=W\cdot[1,1-\alpha,\alpha]\; ,\qquad |\cP_4|=24\; .
\end{split}
\end{align}
Thus $\cP:=\cP_1\sqcup\cP_2\sqcup\cP_3\sqcup\cP_4$ contains $|\cP|=3+4+12+24=43$ rays.

The orbit cardinalities in \eqref{eq: ray orbits} follow directly from the signed-permutation action. For instance, $[1,1,1]$ is determined by its sign pattern modulo an overall sign and therefore has $2^{3}/2=4$ projective images. For $[1,1-\alpha,\alpha]$, the absolute values of the three coordinates are respectively $1,\sqrt3,\sqrt2$, and hence are distinct; there are therefore six coordinate permutations and four projectively distinct relative sign patterns, giving $6\cdot4=24$ rays.

Next, we take four $W$-orbits of contexts, equivalently orthogonal bases, represented by
\begin{align}\label{eq:CD}
\begin{split}
    C_1
    &=\bigl\{[1,0,0],[0,1,0],[0,0,1]\bigr\}\; ,\\
    C_2
    &=\bigl\{[1,0,0],[0,1,\alpha],[0,\alpha,1]\bigr\}\; ,\\
    C_3
    &=\bigl\{[1,1,1],[1,-\alpha,\alpha-1],[1-\alpha,\alpha,-1]\bigr\}\; ,\\
    C_4
    &=\bigl\{[1,\alpha,0],[\alpha,1,1-\alpha],[\alpha,1,\alpha-1]\bigr\}\; ,\
\end{split}
\end{align}
Since $\overline{\alpha}=-\alpha$ and $\alpha^{2}=-2$, all four representatives are orthogonal bases. Their $W$-orbits will be denoted $\cC_1=W\cdot C_1$, $\cC_2=W\cdot C_2$, $\cC_3=W\cdot C_3$ and $\cC_4=W\cdot C_4$, and have cardinalities
\begin{equation*}
    |\cC_1|=1\; ,\qquad
    |\cC_2|=6\; ,\qquad
    |\cC_3|=12\; ,\qquad
    |\cC_4|=12\; .
\end{equation*}
Consequently the arrangement contains $|\cC|=1+6+12+12=31$ contexts.

\begin{theorem}\label{thm: qutrit LCS}
    The LCS over $\zz_3$ associated with the $43$ ray-$31$ context arrangement $(\cP,\cC)$ above is classically unsatisfiable.
\end{theorem}

\begin{proof}
    The incidence pattern is determined by the four ray and context orbits. Clearly, the coordinate context consists of the three $\cP_1$-rays. A single $\cC_2$-context contains one $\cP_1$-ray and two $\cP_3$-rays, hence, every $\cP_1$-ray occurs in $\frac{|\cC_2|}{|\cP_1|}=2$, and every $\cP_3$-ray in $\frac{2|\cC_2|}{|\cP_3|}=1$ such contexts. Similarly, a $\cC_3$-context contains one $\cP_2$-ray and two $\cP_4$-rays, hence, every $\cP_2$-ray occurs in $\frac{|\cC_3|}{|\cP_2|}=3$, and every $\cP_4$-ray in $\frac{2|\cC_3|}{|\cP_4|}=1$ such contexts. Finally, every $\cC_4$-context contains one $\cP_3$-ray and two $\cP_4$-rays, hence, every $\cP_3$-ray occurs in $\frac{|\cC_4|}{|\cP_3|}=1$, and every $\cP_4$-ray in $\frac{2|\cC_4|}{|\cP_4|}=1$ such contexts. In summary:
    \begin{center}
    \begin{tabular}{c|cccc}
        \toprule
        & $\cC_1$ & $\cC_2$ & $\cC_3$ & $\cC_4$ \\
         \midrule
        $\cP_1$ & $1$ & $2$ & $0$ & $0$\\
        $\cP_2$ & $0$ & $0$ & $3$ & $0$\\
        $\cP_3$ & $0$ & $1$ & $0$ & $1$\\
        $\cP_4$ & $0$ & $0$ & $1$ & $1$\\
        \bottomrule
    \end{tabular}
    \end{center}
    Now, add all constraint equations $\sum_{v\in C}x_v\equiv 1\pmod 3$ belonging to $\cC_1$, $\cC_2$ and $\cC_3$, and subtract those belonging to $\cC_4$. For a ray in $\cP_1$ the resulting coefficient is $1+2=3\equiv0\pmod3$, for a ray in $\cP_2$ it is $3\equiv0\pmod3$, and for rays in $\cP_3$ and $\cP_4$ it is $1-1=0$. On the other hand, $|\cC_1|+|\cC_2|+|\cC_3|-|\cC_4|=7\equiv1\pmod3$, which contradicts the earlier count.
\end{proof}

Equivalently, if $A$ denotes the $31\times43$ context--ray incidence matrix, Prop.~\ref{prop: cycle criterion} supplies a vector $\lambda\in\mathbb F_3^{31}$ which is constant on the four $W$-orbits of contexts, with orbit coefficients $(+1,+1,+1,-1)$, explicitly $\lambda=\sum_{C\in\cC_1}[C]+\sum_{C\in\cC_2}[C]+\sum_{C\in\cC_3}[C]-\sum_{C\in\cC_4}[C]$ satisfies $A^{T}\lambda=0\pmod p$ and $\mathbf{1}^T\lambda=1\pmod p$. Finally, a quantum solution exists by Prop.~\ref{prop: canonical qsol}, hence, the LCS associated with the ray-context arrangement above provides a quantum--classical gap for the pair $(n=3,d=3)$.

\subsection{Quantum solutions to LCS over $\zz_n$ for arbitrary $2\leq n\in\mathbb{N}$}\label{sec: full solution}

We now consider LCS over $\zz_n$ for arbitrary integers $2\leq n\in\mathbb{N}$.

\begin{lemma}\label{lm: prime divisor lift}
    Let $A$ be the incidence matrix of a ray-context arrangement in $\C^d$ (see Def.~\ref{def: canonical LCS}), and let $p\mid n$ be prime. If $Ax\equiv\mathbf{1}\pmod p$ is classically unsatisfiable, then so is $Ax\equiv\mathbf{1}\pmod n$.
\end{lemma}

\begin{proof}
    Recall from Prop.~\ref{prop: cycle criterion} that $Ax\equiv\mathbf{1}\pmod p$ is classically unsatisfiable if and only if there exists an incidence cycle over $\Fp$ of the form $A^T\lambda\equiv0\pmod p$ and $\mathbf1^T\lambda\not\equiv0\pmod p$. Choose integer representatives of the entries of $\lambda$ such that $A^T\lambda=pz$ for some integer vector $z$. Then $A^T\lambda'=nz\equiv 0\pmod n$ for $\lambda'=\frac{n}{p}\lambda$, and since $p\nmid\mathbf1^T\lambda$, we also have $\mathbf1^T\lambda'=\frac{n}{p}\mathbf1^T\lambda\not\equiv0\pmod n$.
\end{proof}

We thus obtain the following general existence result for LCS with a quantum-classical gap.

\begin{theorem}\label{thm: qsols general}
    Let $2\leq n\in\mathbb{N}$ and $d\in\mathbb{N}$. There exists a linear constraint system over $\zz_n$ which has no classical solution but admits a $d$-dimensional quantum solution if and only if
    \begin{align*}
        \gcd(n,d)>1
        \qquad\text{and}\qquad
        d\geq 3\; .
    \end{align*}
    In particular, the minimal dimension of a quantum solution of the form in Prop.~\ref{prop: canonical qsol} to a classically unsatisfiable LCS over $\zz_n$ is the least prime divisor of $n$ for $n$ odd, $3$ if $3\mid n$ and $4$ if $2\mid n$, $3\nmid n$.
\end{theorem}

In other words, the determinant condition $\gcd(n,d)>1$ in Prop.~\ref{prop: determinant} is sufficient in every dimension except for $d=2$, where quantum and classical solvability coincide.

\begin{proof}
    Suppose that a LCS over $\zz_n$ has a $d$-dimensional quantum solution $\{X_i\}_i$. Taking determinants with $\zeta^{s_i}_n=\det(X_i)$ gives $As\equiv db\pmod n$, hence, if $\gcd(n,d)=1$, then $d^{-1}s$ is a classical solution. This implies $\gcd(n,d)>1$. Moreover, if $d=2$, then any two commuting non-scalar $2\times2$ unitaries share the same two eigenvectors, hence, all variables connected via some relation admit a common eigenbasis, and the eigenvalues of either eigenspace satisfy the LCS. This implies $d\geq 3$.
    
    Conversely, suppose $\gcd(n,d)>1$ and $d\geq 3$. If $p$ odd prime divides $\gcd(n,d)$, a classically unsatisfiable LCS over $\Fp$ with a $p$-dimensional quantum solution exists by Thm.~\ref{thm: d=p>3 prime} for $p\geq 5$ and by Thm.~\ref{thm: qutrit LCS} for $p=3$. Moreover, Lm.~\ref{lm: prime divisor lift} yields a classically unsatisfiable LCS over $\zz_n$. After tensoring every operator with $\one_k$ for $d=kp$ thus yields a $d$-dimensional quantum solution.
    
    It remains to consider the case in which no odd prime divides $\gcd(n,d)$. In this case, $n$ and $d$ are both even, and since $d\geq 3$, this implies $d\geq 4$. There exist rank-one parity Kochen--Specker arrangements in every (even) dimension $d\geq 4$ (see e.g. Ref.~\cite{WaegellAravind2017}), their canonical LCS are classically unsatisfiable over $\mathbb F_2$, and Lm.~\ref{lm: prime divisor lift} thus yields a classically unsatisfiable LCS over $\zz_n$. Its canonical projector operators give a quantum solution directly in dimension $d$ by Prop.~\ref{prop: canonical qsol}.
\end{proof}

\textbf{$\zz_n$-valued frame functions.} Using similar reasoning as in Ref.~\cite{HardingJagerSmith2005}, Thm.~\ref{thm: frame functions vs qsols} and Thm.~\ref{thm: qsols general} yield an improvement on the characterisation of $\zz_n$-valued frame functions in Thm.~3 of Ref.~\cite{HardingJagerSmith2005}.

\begin{lemma}[\cite{HardingJagerSmith2005}]\label{lm: dimension lifting}
    If every $\zz_n$-valued frame function on $\mathbb C^{d'}$ for $d'\geq 2$ has weight zero, then every $\zz_n$-valued frame function on $\mathbb C^d$ is constant on rank-one projectors for every $d\ge d'+1$.
\end{lemma}

\begin{proof}
    Let $f$ be an $\zz_n$-valued frame function on $\mathbb C^d$, and let $P,Q\in\cP_1(\C^d)$ be any two rank-one projectors. Since $d\geq d'+1$, there are $d'-1$ mutually orthogonal rank-one projectors $R_1,\ldots,R_{d'-1}$ perpendicular to both $P$ and $Q$. By assumption, the restriction of $f$ to the $d'$-dimensional subspaces $P\oplus R_1\oplus\cdots\oplus R_{d'-1}$ and $Q\oplus R_1\oplus\cdots\oplus R_{d'-1}$ has weight zero, hence,
    \begin{align*}
        f(P)+\sum_{i=1}^{d'-1}f(R_i)
        =0
        =f(Q)+\sum_{i=1}^{d'-1}f(R_i)\; ,
    \end{align*}
    and therefore $f(P)=f(Q)$.
\end{proof}

\begin{theorem}\label{thm: frame function rigidity}
    Let $n\ge2$. Suppose that
    \begin{align*}
        d>p
        \quad\text{for every odd prime }p\mid n,
    \end{align*}
    and, if $2\mid n$, suppose in addition that $d\ge5$. Then every $\zz_n$-valued frame function on $\mathbb C^d$ is constant.
    
    Consequently, its possible weights are of the form $\gcd(n,d)\zz_n$.
\end{theorem}

\begin{proof}
    Let $p\mid n$ be prime. For odd $p$, the LCS in Thm.~\ref{thm: d=p>3 prime} and Thm.~\ref{thm: qutrit LCS} imply that every $\zz_p$-valued frame function on $\C^p$ has weight zero, as a frame function of nonzero weight $w\in\zz_p$ could be rescaled to give a classical solution of the canonical equations by Lm.~\ref{lm: ff + invertible weight -> csol}. Lm.~\ref{lm: dimension lifting} therefore implies that every $\zz_p$-valued frame function is constant whenever $d>p$. For $p=2$, the $18$ ray-$9$ context refinement of the Mermin-Peres square in dimension $4$ similarly forces every $\zz_2$-valued frame function to have weight zero, and Lm.~\ref{lm: dimension lifting} thus implies constancy for $d\geq 5$.
    
    It remains to pass from prime coefficients to arbitrary $n$. Harding, Jager and Smith show that constancy for cyclic coefficients lifts to their powers by induction \cite[Cor.~4]{HardingJagerSmith2005}; the same statement applies here with the improved prime-dimensional input above. Thus every $\zz_{p^a}$-valued frame function is constant for every prime power $p^a\Vert n$.\footnote{Here, $p^a\Vert n$ means that $p^a\mid n$ but $p^{a+1}\nmid n$, that is, $p^a$ is the exact power of $p$ dividing $n$.} The Chinese remainder theorem then implies that every $\zz_n$-valued frame function is constant on rank-one projectors, $f(P)=c$; equivalently, its additive extension satisfies $f(Q)=c\cdot\mathrm{rank}(Q)$, and its weight is therefore $dc$. Consequently, the set of possible weights is the image of multiplication by $d$ in $\zz_n$, namely $d\zz_n=\gcd(n,d)\zz_n$.
\end{proof}

\subsection{Quantum solutions to LCS over arbitrary finite fields $\Fq$}\label{sec: finite fields}

There is another natural generalisation of LCS over $\zz_p$ beyond the case of $p$ prime, in terms of finite fields. For $q=p^f$, let  $A\in\mathrm{Mat}_{M\times N}(\Fq)$ and $b\in\Fq^M$. Classically, a solution is $x\in\Fq^N$ with $Ax=b$ over $\Fq$. Note that since the field $\Fq$ has characteristic $p$, that is, $pa=0$ for every $a\in\Fq$, its additive group is the elementary Abelian group $(\zz_p)^f$. For $f>1$ this group is not cyclic, hence, there is no single unitary whose powers can faithfully encode all field values.

Instead, define the map $\tr_{\Fq/\Fp}:\mathbb{F}_q\ra\mathbb{F}_p$ by $\tr_{\Fq/\Fp}(x):=\sum_{i=0}^{f-1}x^{p^i}$, and fix the additive character
\begin{equation}\label{eq: additive character}
    \psi(a)=\exp\left(\frac{2\pi i}{p}\tr_{\Fq/\Fp}(a)\right)\; .
\end{equation}
Importantly, note that the trace pairing $(s,a)\mapsto\tr_{\mathbb F_q/\mathbb F_p}(sa)$ is nondegenerate, consequently every additive character of $\Fq$ is uniquely of the form $a\mapsto\psi(sa)$ for some $s\in\Fq$.

\begin{definition}\label{def: Fqsol}
    A quantum solution of $Ax=b$ over $\Fq$ on $\cH=\C^d$ is a family of unitary representations $X_i:(\Fq,+)\to\UH$ such that variables occurring in a common row have commuting images and such that, with $\psi$ as in Eq.~(\ref{eq: additive character}), for every row $r\in[M]$ and every $t\in\Fq$,
    \begin{equation}\label{eq: Fq constraints}
        \prod_{i=1}^NX_i(tA_{ri})
        =\psi(tb_r)\one\; .
    \end{equation}
\end{definition}

For $q=p$ the additive group is cyclic and $X_i$ is determined by the single unitary $X_i(1)$, that is, Def.~\ref{def: Fqsol} reduces to the definition of quantum solutions in Eq.~(\ref{eq: qsol}). In turn, for $f>1$ Def.~\ref{def: Fqsol} differs from Eq.~\ref{eq: qsol} which instead has $X^q=\one$ and thus models a cyclic group $\zz_q$ of exponent $q$. By contrast, $(\Fq,+)$ has exponent $p$ and is noncyclic for $f>1$. As a consequence $\psi$ is not injective, and the quantifier in Eq.~(\ref{eq: Fq constraints}) is necessary to ensure that the operator relation faithfully encodes the underlying additive relation, that is, $\psi(ta)=1\ \forall t\in\Fq\Leftrightarrow a=0$, while $\psi(a)=1\not\Leftrightarrow a=0$.

With Def.~\ref{def: Fqsol}, Prop.~\ref{prop: canonical qsol} is readily generalised. For a rank-one projector $P\in\cP_1(\C^d)$ define
\begin{equation}\label{eq:Fqprojector}
    X_P(a)
    =(\one-P)+\psi(a)P\; .
\end{equation}
Then $X_P(a+b)=X_P(a)X_P(b)$. Over contexts $C$ (resolutions of the identity), this gives
\begin{align*}
    \prod_{P\in C}X_P(t)
    =\psi(t)\one\qquad\forall t\in\Fq\; ,
\end{align*}
hence, every LCS over $\Fq$ associated to a ray-context arrangement has a quantum solution.

\begin{lemma}\label{lm: finite field extension}
    Let $q=p^f$. Suppose that an LCS $Ax=b$ over $\Fp$ with $A\in\mathrm{Mat}_{M\times N}(\Fp)$ and $b\in\Fp^M$ has no classical solution but admits a $d$-dimensional quantum solution. Then the same LCS has no classical solution over $\F_q$ yet admits a $d$-dimensional quantum solution in the sense of Def.~\ref{def: Fqsol}.
\end{lemma}

\begin{proof}
    By Prop.~\ref{prop: cycle criterion}, classical inconsistency is detected by a chain $\lambda\in\ker(A^T)\subset\Fp^M\cong\cC_1$ such that
    \begin{align*}
        A^T\lambda=0\; ,\qquad
        b^T\lambda\neq0\; .
    \end{align*}
    These relations remain valid in $\Fq$ under the natural inclusion $\Fp\subset\Fq$, hence, $Ax=b$ remains classically unsatisfiable over $\Fq$.
    
    Let $X_1,\ldots,X_N$ be a quantum solution to the LCS over $\Fp$, and let $X_i=\sum_{a\in\Fp}\zeta^a_p P_{i,a}$ with $\zeta_p=e^{\frac{2\pi i}{p}}$ be its spectral decomposition. For $s\in\Fq$, define
    \begin{align*}
        X_i(s)
        :=\sum_{a\in\Fp}\psi(sa)P_{i,a}.
    \end{align*}
    Since $\psi$ is additive, $X_i(s+s')=X_i(s)X_i(s')$, hence, $X_i$ defines a unitary representation of $(\Fq,+)$. Moreover, commuting $X_i$'s have commuting spectral projections, hence, variables occurring in a common row have commuting representation images.
    
    Now, fix a row $r\in[M]$, choose a common rank-one spectral resolution $\{P_k\}_{k=1}^d$ for the operators occurring in that row, and let $a_{ik}\in\Fp$ be the exponent of $X_i$ on the $k$-th common eigenspace. The original operator equation implies $\sum_iA_{ri}a_{ik}=b_r$ in $\Fp$ for every $k$. Hence, for every $t\in\Fq$,
    \begin{equation*}
        \prod_iX_i(tA_{ri})
        =\prod_i\sum_{k=1}^d\psi(tA_{ri}a_{ik})P_k
        =\sum_{k=1}^d\psi\left(t\sum_iA_{ri}a_{ik}\right)P_k
        =\psi(tb_r)\sum_{k=1}^dP_k
        =\psi(tb_r)\one\; .\qedhere
    \end{equation*}
\end{proof}

Lm.~\ref{lm: finite field extension} applies to quantum solutions of any type, in particular, it applies to incidence data of ray-context arrangements as well as tensorised quantum solutions in the binary case.

\begin{theorem}\label{thm: qsols ff}
    For every finite field $\Fq$ with $q=p^f$ for $p$ prime, there exists a LCS with no classical solution, yet with a $d$-dimensional quantum solution if and only if $p\mid d$ and $d\geq 3$.
\end{theorem}

\begin{proof}
    Suppose first that $X_i:(\Fq,+)\to\U(\C^d)$ is a quantum solution, and note that for each $i$, the determinant $\chi_i(a):=\det(X_i(a))$ is an additive character of $\Fq$. Since $\psi$ in Eq.~(\ref{eq: additive character}) is nontrivial, nondegeneracy of the trace pairing implies that there is a unique $s_i\in\Fq$ such that $\chi_i(a)=\psi(s_i a)$ for all $a\in\Fq$. Taking determinants in the row relations gives, for every $t\in\Fq$,
    \begin{align*}
        \psi\left(t\sum_i A_{ri}s_i\right)
        =\psi(t d b_r)\; ,
    \end{align*}
    hence, $As=db$. If $p\nmid d$, then $d$ is invertible in $\Fq$ and $d^{-1}s$ is a classical solution.
    
    The case $d=2$ is excluded by similar reasoning as in Thm.~\ref{thm: qsols general}: restriction of each representation to common eigenspace yields additive characters $a\mapsto\psi(s_ia)$, hence, the row relations read $\psi(t(\sum_iA_{ri}s_i-b_r))=1$ for all $t\in\Fq$, which is equivalent to $\sum_iA_{ri}s_i=b_r$ in $\Fq$.
    
    Conversely, suppose $p\mid d$ and $d\geq 3$. By Thm.~\ref{thm: qsols general}, there exists a classically unsatisfiable LCS over $\Fp=\zz_p$ with a $d$-dimensional quantum solution. Lm.~\ref{lm: finite field extension} lifts the same LCS and quantum solution to $\Fq$, while preserving classical unsatisfiability.
\end{proof}

\section{Conclusion}\label{sec: conclusion}

We completely resolve the existence problem of linear constraint systems (LCS) that have no classical solution, that is, are inconsistent, yet admit finite-dimensional quantum solutions in the form of unitary operators on a finite-dimensional Hilbert space. Such examples were known previously in the binary case, with quantum solutions built from $n$-qubit Pauli operators.$^{\ref{fn: prior work}}$

We show that, unlike the qubit Pauli-based constructions in the binary case, a similar tensor-product decomposition with higher-dimensional qudit tensor factors of the unitary operators in a quantum solution never corresponds to a classically unsatisfiable LCS - unless one of its tensor factors already furnishes such an example. We then considered LCS associated to ray-context arrangements, that is, sets of rays in and their resolutions of the identity. Such arrangements define quantum solutions to their naturally associated LCS, shifting the problem of a quantum-classical gap for such LCS to finding arrangements whose incidence relations are inconsistent.

We explicitly relate the existence of quantum solutions to such LCS with $\zz_n$-valued frame functions, which allows us to obtain a family of LCS with a quantum-classical gap from the work of Ref.~\cite{HardingJagerSmith2005}. Moreover, we complement their construction to obtain a general and constructive existence result for LCS with quantum but no classical solution in the variables $n,d$. In particular, we give a $43$ ray-$31$ context arrangement in $\C^3$, whose associated LCS is classically unsatisfiable.\\

\textbf{Acknowledgements.} OpenAI GPT-5.6 Sol was used as an interactive aid during exploratory computational searches and proof development. All mathematical statements and exact computations reported here were subsequently checked by the author.

\bibliographystyle{abbrv}
\bibliography{bibliography}

\appendix

\section{Relation with Peres-Penrose qutrit contextuality proofs}\label{app: d=p=3}

We consider the construction in Thm.~\ref{thm: qutrit LCS} from the perspective of Kochen-Specker contextuality \cite{KochenSpecker,Specker1960}. Since a valuation, that is, a (noncontextual) $0,1$-valued assignment to rank-one projectors such that in every resolution exactly one projector is assigned the value $1$ and all others $0$, is in particular a ${0,1}$-valued solution of the constraint equations $\sum_{P\in C}x_P\equiv 1\pmod 3$, the ray--context arrangement in Thm.~\ref{thm: qutrit LCS} is also KS-uncolourable. (More generally, any ray-context arrangement whose associated LCS is classically unsatisfiable does not admit valuations.)

In fact, the occurrence of the ring $\zz[\sqrt{-2}]$ is closely related to the classical Peres construction \cite{Peres1991}. Peres' $33$-ray proof of the Kochen--Specker theorem uses the real alphabet
\begin{align*}
    \{0,\pm1,\pm\sqrt2\}
\end{align*}
and is organised into $16$ orthogonal triads. Gould and Aravind subsequently showed that the orthogonality graph of the Peres set is isomorphic to that of Penrose's complex $33$ ray construction \cite{Penrose1994}, and that the two belong to a continuous family of unitarily inequivalent realisations of the same orthogonality graph \cite{GouldAravind2010}. More recently, systematic searches over algebraic coordinate alphabets have shown that $\zz[\sqrt{-2}]$ supports a $33$-ray realisation of the same Peres-type orthogonality graph \cite{Kernaghan2026}. The underlying cancellation mechanism is the same norm-$2$ identity, $\alpha\bar\alpha=2$, which plays the role of $(\sqrt2)^2=2$ in the real coordinates of Peres' construction.

Recent work further relates these algebraic realisations with the deformation family of the Peres--Penrose graph \cite{FloresGordillo2026}. For the present comparison it is useful to choose the following
$W=W(B_3)$-symmetric realisation of the
$\zz[\sqrt{-2}]$ Peres orthogonality graph:
\begin{equation*}
    \cP_{33}^{(-2)}
    =W\cdot[1,0,0]\sqcup W\cdot[0,1,-1]\sqcup W\cdot[0,1,\alpha]\sqcup W\cdot[1,1,\alpha]\; .
\end{equation*}
The four orbit sizes are $3,6,12,12$, giving $33$ rays in total. Direct enumeration gives $72$ orthogonal pairs and $16$ complete triads, in agreement with the Peres orthogonality graph.

The intersection of $\cP_{33}^{(-2)}$ with the present arrangement $\cP$ is particularly simple:
\begin{align}\label{eq: Peres seed rays}
    \cP_{15}
    :=\cP_{33}^{(-2)}\cap\cP
    =W\cdot[1,0,0]\sqcup W\cdot[0,1,\alpha]\; .
\end{align}
The two constructions thus have exactly $3+12=15$ rays in common, and these rays support seven common orthogonal bases: namely, the coordinate basis and the six $W$-images of
\begin{equation}\label{eq: Peres seed contexts}
    \bigl\{[1,0,0],[0,1,\alpha],[0,\alpha,1]\bigr\}\; .
\end{equation}
Consequently, the two geometries may be viewed as different symmetric extensions of the same $15$ ray-$7$ context skeleton:
\begin{align}
    \cP_{33}^{(-2)}
    &=\cP_{15}\sqcup W\cdot[0,1,-1]\sqcup W\cdot[1,1,\alpha]\; ,\nonumber\\
    \cP
    &=\cP_{15}\sqcup W\cdot[1,1,1]\sqcup W\cdot[1,1-\alpha,\alpha]\; .\label{eq: Peres' LCS}
\end{align}
The first branch adds $6+12=18$ rays and produces the $33$-ray Peres-type geometry, whereas the second adds $4+24=28$ rays and produces the present $43$ ray-$31$ context arrangement.

There is a second useful comparison. In the usual $33$ ray description of the Peres proof, orthogonal pairs that do not belong to one of the explicitly displayed $33$ ray triads still enter the KS colouring argument. If these pairs are all completed explicitly to three-element contexts, one obtains a $57$ ray-$40$ context representation of the Peres geometry \cite{Pavicic2019}. Over $\zz[\sqrt{-2}]$ this completion can be written especially simply as
\begin{equation}\label{eq: Peres57}
   \cP_{57}^{(-2)}
    =\cP_{33}^{(-2)}\sqcup W\cdot[3,-1,-\alpha]\; ,
\end{equation}
where the additional orbit contains $24$ rays. The resulting configuration has $57$ rays and $40$ triads. The $24$ ray orbit in Eq.~(\ref{eq: Peres57}) should not be confused with the $24$ ray orbit $W[1,1-\alpha,\alpha]$ in Eq.~(\ref{eq: Peres' LCS}), as they are different projective $W(B_3)$-orbits. Hence, the present $43$ ray-$31$ context arrangement is not obtained simply by deleting $14$ rays from the completed Peres $57$ ray configuration. Rather, it is a genuinely different symmetric extension of the common skeleton Eq.~(\ref{eq: Peres seed rays}) and Eq.~(\ref{eq: Peres seed contexts}).

This distinction is also visible over $\zz_3$. If one retains only the $16$ complete triads of $\cP_{33}^{(-2)}$, their $16\times33$ incidence matrix has full row rank over $\zz_3$. The same is true for the $40\times57$ incidence matrix of $\cP_{57}^{(-2)}$, that is, $\mathrm{rank}_{\mathbb{F}_3}(A_{33})=16$ and $\mathrm{rank}_{\mathbb{F}_3} (A_{57})=40$. Consequently, the Peres incidence system does not possess a nontrivial $\zz_3$ row relation among the constraint equations of its associated LCS.\footnote{Consistency of this LCS over $\mathbb{F}_3$ is also implied by the $\{0,\frac{1}{2},1\}$-colouring of its 57 rays in Ref.~\cite{Ramanathan2024}.} More generally, ordinary KS uncolourability therefore does not necessarily produce a classical obstruction. By
contrast, the present $43$ ray-$31$ context arrangement has
\begin{align*}
   \mathrm{rank}_{\mathbb{F}_3}(A_{43})
   =30\; ,\qquad
   \mathrm{rank}_{\mathbb{F}_3}[A_{43}\mid{\bf1}]
    =31\; ,
\end{align*}
and the missing row relation is precisely the analytic four-orbit certificate of Thm.~\ref{thm: qutrit LCS}.\\

\textbf{Arithmetic interpretation.} The comparison above suggests a simple arithmetic explanation for why the present construction appears with $\zz[\sqrt{2}]$ replaced by $\zz[\sqrt{-2}]$. Peres' arrangement uses the norm identity $\alpha\bar\alpha=2$, which is the complex analogue of $(\sqrt2)^2=2$ and is exactly the norm-$2$ cancellation mechanism identified in recent algebraic studies of three-dimensional KS sets \cite{Kernaghan2026}.

The present construction exploits, in addition, the trace-zero identity  $\alpha+\bar\alpha=0$ in $\zz[\sqrt{-2}]$. For example, the entire context orbit $\cC_2$ in Sec.~\ref{sec: d=p=3} is generated by the elementary orthogonality
\begin{align*}
    \langle(0,1,\alpha),(0,\alpha,1)\rangle
    =\alpha+\bar\alpha
    =0\;.
\end{align*}
The trace-zero identity is absent at the ordinary real Peres point $\alpha=\sqrt2$ and is not implied merely by the condition $|\alpha|^{2}=2$. The point $\alpha=\sqrt{-2}$ therefore carries additional orthogonality relations unavailable at a generic point of the Peres--Penrose deformation family.

In this sense, the $43$ ray-$31$ context arrangement may be regarded as an arithmetic enhancement of the Peres-type geometry with coordinates in $\zz[\sqrt{-2}]$ instead of $\zz[\sqrt{2}]$: the familiar norm-$2$ cancellation produces the underlying Peres-like skeleton, while the additional relation $\alpha+\bar\alpha=0$ permits a different $W(B_3)$-symmetric completion. It is precisely this extra arithmetic and incidence structure that turns ordinary Kochen--Specker uncolourability into the stronger obstruction over $\zz_3$.

\section{A family of LCS from conference matrices}\label{app: conference LCS}

We provide another family (in addition to the HJS construction in Sec.~\ref{sec: HJS}) of LCS over $\zz_p$ with $p$ odd prime that are classically unsatisfiable, yet admit quantum solutions in dimension $d=kp$ for $k\geq 2$.

\subsection{Conference matrices and their sign orbits}\label{sec: conference frames}

A real \emph{conference matrix of order $n$} is a matrix $C\in\mathrm{Mat}_n(\R)$ with $C_{ii}=0$, $C_{ij}\in\{\pm1\}$ and
\begin{equation}\label{eq: conference matrix}
    CC^T=(n-1)\one_n\; .
\end{equation}
Its rows form an orthogonal frame, each row has support $n-1$ and squared norm $n-1$. Paley's construction provides a real conference matrix of order $q+1$ for every odd prime power $q$ \cite{Paley}, in particular, for every odd prime $p$ it provides a real conference matrix of order $p+1$.

For completeness, we recall his construction explicitly. The multiplicative group $\Fp^\times$ is cyclic of order $p-1$. Its subgroup of nonzero squares therefore has index two, and the quotient character
\begin{align*}
    \chi:\Fp^\times\longrightarrow\{\pm1\}
\end{align*}
is given by $\chi(x)=1$ on squares and $\chi(x)=-1$ on nonsquares. $\chi$ is a group homomorphism, hence, $\chi(xy)=\chi(x)\chi(y)$ which remains valid after extending $\chi$ by $\chi(0)=0$. Moreover, since there are as many squares as nonsquares, one finds
\begin{equation}\label{eq: zero character sum}
    \sum_{t\in\Fp}\chi(t)
    =0\; .
\end{equation}
For $c\neq0$, consider $S(c)=\sum_{t\in\Fp}\chi(t)\chi(t+c)$. Clearly, the term for $t=0$ vanishes. For $t\neq 0$, let $u=(t+c)/t=1+c/t$, which  defines a bijection from $\Fp^\times$ onto $\Fp\setminus\{1\}$, and we therefore have
\begin{align*}
    \chi(t)\chi(t+c)
    =\chi(t^2u)=\chi(u)\; .
\end{align*}
Hence, using Eq.~(\ref{eq: zero character sum}), $S(c)=\sum_{u\neq1}\chi(u)=-1$, and thus
\begin{equation}\label{eq: Paley character sum}
    \sum_{t\in\Fp}\chi(x-t)\chi(y-t)=
    \begin{cases}
        p-1,&x=y\\
        -1,&x\neq y
    \end{cases}\; ,
\end{equation}
Let $Q$ be the $p\times p$ matrix indexed by $x,y\in\Fp$ with $Q_{xy}=\chi(x-y)$. Eq.~(\ref{eq: Paley character sum}) gives $QQ^T=p\one_p-J_p$, where $J_p=\mathbf{1}\mathbf{1}^T$ is the all-$1$ matrix in $\mathrm{Mat}_p(\R)$. Now, let $s=\chi(-1)$, and add one additional row and column to define the augmented matrix $C_p=\begin{pmatrix}0&\mathbf1^T\\ s\mathbf1&Q\end{pmatrix}$. Then one verifies
\begin{align*}
    C_pC_p^T
    =\begin{pmatrix}
        p&\mathbf1^TQ^T\\
        Q\mathbf1&s^2J_p+QQ^T
    \end{pmatrix}\; ,
\end{align*}
and since $Q\mathbf1=0$ by Eq.~(\ref{eq: zero character sum}) and $s^2=1$ we find $C_pC_p^T=p\one_{p+1}$.

Fix a coordinate block $S$ of size $p+1$. The projective diagonal-sign group is
\begin{align*}
    E_S
    =\{\pm1\}^{S}/\{\pm(1,\ldots,1)\}\cong(\zz_2)^p\; .
\end{align*}
Since $p+1$ is even, the sign product is invariant under multiplication by the global sign. Hence,
\begin{equation}\label{eq:Hpsign}
    H_S
    =\left\{[\varepsilon]\in E_S\mid\prod_{j\in S}\varepsilon_j=1\right\}
\end{equation}
is a well-defined index-two subgroup of order $2^{p-1}$.

\begin{lemma}\label{lm: sign orbit}
    The action of $H_S$ on any conference frame is free.  Its orbit contains every projective support-$p$ signed ray in the block $S$ exactly once.  Consequently $2^{p-1}$ is the minimum number of conference frames needed to cover all such rays.
\end{lemma}

\begin{proof}
    Column sign changes preserve orthogonality because a diagonal sign matrix is orthogonal. Each conference row has a unique zero coordinate. If a projective sign transformation stabilises the whole frame, it cannot permute rows with different zero positions; each row must therefore be fixed projectively. Comparing its nonzero coordinates for different rows forces all column signs to be equal, hence the projective transformation is trivial. The $H_S$-orbit therefore has $2^{p-1}$ frames.
    
    Fix a support-$p$ projective signed ray with zero coordinate $i$. The signs on its $p$ nonzero coordinates determine the required column signs up to an overall sign. The remaining sign at coordinate $i$ does not affect the ray and can be chosen uniquely so that the total sign product is $+1$.  Thus exactly one element of $H_S$ maps the $i$th conference row to the prescribed ray.  There are $(p+1)2^{p-1}$ support-$p$ projective signed rays in the block, and every frame contains $p+1$ of them, so no smaller family of frames can cover all rays.
\end{proof}

The larger sign group $E_S$ gives twice as many frames, and the full hyperoctahedral closure (see Eq.~(\ref{eq: projective Gp})) is larger still. Lm.~\ref{lm: sign orbit} shows that the reduced group $H_S$ is the natural minimal choice if one wishes to retain every support-$p$ signed ray in each block.

\subsection{Explicit construction}\label{sec: universal construction primes}

For $p$ odd prime, consider the following ray-context arrangement in $\R^d$ with $d=kp$ and $k\geq 2$. As before, we denote by $[v]$ the ray, equivalently rank-$1$ projection corresponding to the vector $v\in\R^d$.\\

\textbf{The arrangement.} Take the two projective signed-support classes
\begin{align}\label{eq:Rpgeneral}
    \cP_1&=\{[e_i]\mid 1\le i\le d\}\; , &
    \cP_p&=\{[v]\mid v\in\{0,\pm1\}^d,\ |\supp(v)|=p\}\; .
\end{align}
These define single orbits under the projective hyperoctahedral group
\begin{align}\label{eq: projective Gp}
    \overline{W(B_d)}
    \cong(\zz_2)^{d-1}\rtimes S_d,
    \qquad |\overline{W(B_d)}|=2^{d-1}d!\; .
\end{align}
The cardinalities of their orbits are $|\cP_1|=d$ and $|\cP_p|=2^{p-1}\binom dp$, and those of their stabiliser groups
\begin{align*}
    |\mathrm{Stab}_{\overline{W(B_d)}}([e_1])|
    =2^{d-1}(d-1)!\; ,\qquad
    |\mathrm{Stab}_{\overline{W(B_d)}}([1^p0^{d-p}])|
    =2^{d-p}p!(d-p)!\; .
\end{align*}
Next, consider two types of contexts. The first is the single context corresponding to the coordinate basis $C_0=\{[e_i]\mid i\in[d]\}$. For the second, choose a block $S\subset[d]$ with $|S|=p+1$, place a fixed conference frame of order $p+1$ in $S$ (see Sec.~\ref{sec: conference frames}), and complete the frame by the $d-p-1$ coordinate rays outside $S$. We denote the union over all blocks $S$ of the local $H_S$ orbits by $\cC_1$. The resulting ray and context partitions are equitable (see comment after Prop.~\ref{prop: orbit cycle}), hence, Prop.~\ref{prop: orbit cycle} applies.

To count its elements, note that the symmetric group $S_d$ is transitive on the $\binom{d}{p+1}$ coordinate blocks. A block has stabiliser $S_{p+1}\times S_{d-p-1}$ and orbit size $d!/((p+1)!(d-p-1)!)=\binom{d}{p+1}$. For each fixed block, $H_S$ has order $2^{p-1}$ and Lm.~\ref{lm: sign orbit} shows that the conference-frame stabiliser in $H_S$ is trivial, so its local orbit has size $2^{p-1}$. Consequently, contexts in $\cC_1$ are indexed by pairs $(S,h)$ with $|S|=p+1$ and $h\in H_S$, giving $2^{p-1}\binom{d}{p+1}$ frames. In summary,
\begin{align}\label{eq: Paley contexts}
    |\cP_1\sqcup \cP_p|
    &=d+2^{p-1}\binom{d}{p}\; ,&
    |\{C_0\}\sqcup\cC_1|
    &=1+2^{p-1}\binom{d}{p+1}\; .
\end{align}

\begin{theorem}\label{thm: conference LCS}
    For every odd prime $p$ and $d=kp$ with $k\geq 2$ and $k\not\equiv1\pmod p$, there exists a LCS over $\zz_p$ with a $d$-dimensional quantum solution, yet no classical solution.
\end{theorem}

\begin{proof}
    Consider the canonical LCS associated with the projectors $\cP=\cP_1\sqcup\cP_p$ and contexts $\cC=\cC_0\sqcup\cC_1$ as in the arrangement defined in the paragraphs preceding the theorem. Since conference matrices exist for all prime powers, such an arrangement always exists.

    A quantum solution to this LCS exists by Prop.~\ref{prop: canonical qsol}. To prove that it is inconsistent, we evaluate the criterion in Prop.~\ref{prop: orbit cycle}. We begin to determine the dual incidence matrix in Eq.~(\ref{eq: N}). Its values in the first column are obvious. For the second, note first that a fixed coordinate ray belongs to a conference context exactly when its coordinate is outside the chosen block. Hence, its degree in $\cC_1$ is $2^{p-1}\binom{d-1}{p+1}$. Moreover, a fixed support-$p$ ray in a conference frame is part of a block of size $p+1$, which leaves $d-p$ possible choices.  Within each block Lm.~\ref{lm: sign orbit} gives exactly one conference frame containing the ray, hence, this is the degree of a support-$p$ ray in $\cC_1$. Consequently, we find
    \begin{align*}
        N_d^{(p)}=
        \begin{pmatrix}
            1&2^{p-1}\binom{d-1}{p+1}\\
            0&d-p
        \end{pmatrix}.
    \end{align*}
    Now, define the context chain (over the two-fold decomposition $\cC=\cC_0\sqcup\cC_1$)
    \begin{align*}
        \alpha=(-a_d,1)^T\; ,\qquad
        a_d=2^{p-1}\binom{d-1}{p+1}\; .
    \end{align*}
    Since $p\mid d$, the second row of $N_d^{(p)}\alpha$ is simply $d-p\equiv 0\pmod p$, while the first row vanishes by cancellation. On the other hand, the augmentation of the corresponding context chain is
    \begin{align*}
        \sum_j|\cC_j|\alpha_j
        =-a_d+2^{p-1}\binom{d}{p+1}
        =2^{p-1}\left[\binom{d}{p+1}-\binom{d-1}{p+1}\right]
        =2^{p-1}\binom{d-1}{p}\; ,
    \end{align*}
    by Pascal's rule. Fermat's theorem gives $2^{p-1}\equiv 1\pmod p$, and with Lucas' theorem\footnote{Recall that Lucas' theorem asserts that for $n,r$ with $n=\sum_jn_jp^j$ and $r=\sum_jr_jp^j$ the respective base-$p$ expansions, it holds $\binom nr\equiv\prod_j\binom{n_j}{r_j}\pmod p$.} we find,
    \begin{equation*}
        \binom{kp-1}{p}\equiv k-1\pmod p\; .
    \end{equation*}
    Consequently, $\sum_j|\cC_j|\alpha_j\equiv k-1\pmod p$ is nonzero if and only if $k\not\equiv1\pmod p$.
\end{proof}

For every odd prime $p$, Thm.~\ref{thm: conference LCS} provides an entire family of LCS with quantum but no classical solution. In particular, a quantum solution exists in dimension $d=2p$, for which
\begin{align*}
    N_{2p}^{(p)}&\equiv
    \begin{pmatrix}1&-1\\0&0\end{pmatrix}\pmod p &
    \alpha&=(1,1)^T &
    \sum_j|\cC_j|\alpha_j\equiv 1\pmod p\; .
\end{align*}
We list the cardinalities of the ray-context arrangements with $d=2p$ for $p=3,5,7$ in Thm.~\ref{thm: conference LCS}.
\begin{center}
\begin{tabular}{c@{\qquad}c@{\qquad}r@{\qquad}r}
\toprule
$p$&$d$&rays&contexts\\
\midrule
$3$&$6$&$86$&$61$\\
$5$&$10$&$4042$&$3361$\\
$7$&$14$&$219662$&$192193$\\
\bottomrule
\end{tabular}
\end{center}
Since a conference frame requires $p+1$ coordinates, the family of LCS over $\Fp$ for every odd prime $p$ in Thm.~\ref{thm: conference LCS} yields $d$-dimensional quantum solutions for every $d=kp$ with $k\geq 2$ and $k\not\equiv 1\mod p$.

\section{Quantum solutions from $n$-qubit Pauli operators}\label{app: Pauli comparison}

Binary LCS with quantum but no classical solution are readily constructed from the $n$-qubit Pauli group. In fact, there exist various (partial) characterisations of operator-context arrangements based $n$-qubit Pauli operators \cite{Mermin,Arkhipov2012,TrandafirLisonekCabello2022,OkayRaussendorf2020,MullerGiorgetti2025,AbramskyCercelescuConstantin}. Their incidence-based analysis of the associated LCS being classically unsatisfiable are all special cases of Prop.~\ref{prop: cycle criterion}.\\

\textbf{Homotopical approach.} The incidence complex used here is closely related to two cohomological formulations of contextuality developed elsewhere. In the homotopical approach of Okay and Raussendorf~\cite{OkayRaussendorf2020}, an arrangement of observables and measurement contexts gives rise to a chain complex in which contexts generate degree two chain groups and observable labels generate degree one chain groups, with boundary given by the signed incidence relation of a context. Up to this uniform degree shift, this is precisely the incidence complex used above: our map
\begin{align*}
    \partial
    =A^T:\cC_1^{\mathrm{inc}}\longrightarrow\cC_0^{\mathrm{inc}}
\end{align*}
allows arbitrary LCS coefficients, while the corresponding complex of Ref.~\cite{OkayRaussendorf2020} is obtained from signed context relations. In either formulation a classical value assignment is a cochain whose coboundary equals the prescribed context phase, and contextuality is witnessed by the nontrivial pairing of that cohomology class with an incidence cycle. Consequently,
\begin{align*}
    A^T\lambda=0\; ,\qquad
    b^T\lambda\neq0\; ,
\end{align*}
is the degree-shifted incidence form of the same homological obstruction.\\

\textbf{Cellular cohomology.}  For Pauli realisations there is additional algebraic structure. Let $\cX\subset V$ be any subset of the symplectic Pauli labels and contexts $C\in\cC$ given by subsets of isotropic subspaces $C\subset I\in\mathrm{Iso}(V)$ of the symplectic space $(V,\omega)$ with $V\cong\mathbb{F}^{2n}_2$ and $\omega$ the canonical symplectic (alternating, bilinear, non-degenerate) form given by $\omega(v,w)=\sum_{i=1}^nv_iw_{n+i}-v_{n+i}w_i$.

Note first that the incidence relations of Pauli operators which determine the associated LCS are entirely given by their phase-free symplectic structure $(\cX,\cC)$ with $\cC^\mathrm{inc}_0=\mathbb{F}_2[\cX]$ and $\cC^\mathrm{inc}_1=\mathbb{F}_2[\cC]$.\footnote{The phase-free compatibility relations between $n$-qubit Pauli operators are equivalently encoded by the symplectic form, or by the corresponding anticommutation graph \cite{MullerGiorgetti2025}.} The $n$-qubit Pauli group is a projective representation of $V$, equivalently a central extension
\begin{align*}
    1\longrightarrow\zz_4\longrightarrow\cP_n\longrightarrow V\cong\mathbb{F}^{2n}_2\longrightarrow 1\; .
\end{align*}
The question of which incidence relations admit a realisation in terms of Pauli operators thus additionally involves the choice of Hermitian Pauli representatives $v\mapsto T_v$, that is, a choice of section of this extension. This additional phase freedom is encoded in the LCS via the constraint vector $b\in\mathbb F^{|\cC|}_2$ in $\prod_{v\in C}T_v=(-1)^{b_C}\one$ for all contexts $C\in\cC$. More precisely, since re-phasing does not change the underlying incidence relations, consistency of the associated LCS depends only on the cohomology class $[b]\in\mathrm{coker}(A)=\mathbb{F}^M_2/\im(A)$. In fact, the latter is true for any LCS.

What is special to Pauli-based quantum solutions is that the incidence class $[b]$ is induced by the Pauli multiplication $2$-cocycle $[\beta]$. More precisely, \cite{OkayRobertsBartlettRaussendorf} define a cellular chain complex whose $1$-cells $\cC^\cP_1=\mathbb{F}_2[V]$ are projective Pauli labels and whose $2$-cells encode products of commuting pairs with $\partial_\cP:\cC^\cP_2\ra\cC^\cP_1$ defined by
\begin{align*}
    \partial_\cP[v|w]
    =[w]-[v+w]+[v]\; ,
\end{align*}
where $v\mapsto T_v$ is some choice of Hermitian section with $T_{v+w}=(-1)^{\beta(v,w)}T_vT_w$. Associativity makes $\beta$ a $2$-cocycle, and rephasing the Pauli section changes it by a coboundary.

A context $C\in\cC$ with $C=\{v_1,\cdots,v_l\}$ may be triangulated into multiplication cells. Let $s_j=\sum_{i=1}^jv_i$, where $s_l=0$ encodes the row relation corresponding to the context $C$, and define the two-chain $F_C=\sum_{j=2}^l[s_{j-1}\mid v_j]$ (for any choice of ordering). Then
\begin{align}\label{eq: compatibility}
    \partial_\cP F_C
    =\sum_{j=2}^l([v_j]-[s_j]+[s_{j-1}])
    =\sum_{v\in C}[v]\; ,
\end{align}
The relation between the framework of Ref.~\cite{OkayRobertsBartlettRaussendorf} and our incidence complex with $\partial_\mathrm{inc}=A^T:\mathbb{F}_2[\cC_1]\ra\mathbb{F}_2[\cC_0]$ is thus given as follows: the inclusion $\phi_0:\cC^\mathrm{inc}_0=\mathbb{F}_2[\cX]\ra\cC^\cP_1=\mathbb{F}_2[V]$ together with the map $\phi_1:\cC^\mathrm{inc}_1\ra\cC^\cP_2$ given by $\mathbb{F}_2$-linear extension of the assignment $[C]\mapsto F_C$, define a map $\phi=(\phi_0,\phi_1)$ between chain complexes. Indeed, by Eq.~(\ref{eq: compatibility}), $\phi$ satisfies the (only nontrivial) chain-map identity
\begin{align}\label{eq: chain compatibility}
    \partial_\cP\phi_1
    =\phi_0\partial_\mathrm{inc}\; .
\end{align}
Moreover, evaluating the constraints yields
\begin{align*}
    (-1)^{b_C}\one
    =T_{v_1}\cdots T_{v_l}
    =(-1)^{\beta(s_1,v_2)}T_{s_2}T_{v_3}\cdots T_{v_l}
    =(-1)^{\sum_{j=2}^l\beta(s_{j-1},v_j)}T_{s_l}
    =(-1)^{\sum_{j=2}^l\beta(s_{j-1},v_j)}\one\; ,
\end{align*}
and thus $b_C=\sum_{j=2}^l\beta(s_{j-1},v_j)=\langle\beta,F_C\rangle$. On the other hand, the dual $\phi^*_1:(\cC^\cP_2)^*\ra(\cC^\mathrm{inc}_1)^*$ sends $\beta$ to $\phi^*_1(\beta)([C])=\beta(\phi_1([C]))=\langle\beta,F_C\rangle=b_C$. Hence, the cellular chain $F_\lambda=\phi_1[\lambda]=\sum_C\lambda_CF_C$ of an incidence cycle $\lambda$ satisfies
\begin{align}\label{eq: cohomological pairing}
    b^T\lambda
    =\langle\beta,F_\lambda\rangle
    =\langle[\beta],[F_\lambda]\rangle\; ,
\end{align}
where the last step uses that $\beta$ is a $2$-cocycle and $F_\lambda$ is closed by Eq.~(\ref{eq: chain compatibility}) and $\partial_\mathrm{inc}\lambda=0$, by which the pairing depends only on the corresponding homology
and cohomology classes.

In particular, $[b]\in H^1[\cC^\mathrm{inc}]$ is the pull-back of the cohomology class $[\beta]\in H^2[\cC^\cP]$, and the Pauli cohomological witness may therefore be viewed as a refinement of the general LCS incidence obstruction: the multiplication complex resolves each context into elementary commuting products, while coarse-graining those internal multiplication relations recovers the incidence complex.

\begin{proposition}\label{prop: b <-> beta}
    For a LCS with quantum solution in $\cP_n$, the context phase vector $b$ is the pullback of the Pauli multiplication cocycle $\beta$ to the LCS context complex, that is, 
    $[b]=\phi^*[\beta]$. Moreover, an incidence cycle $\lambda\in\ker A^T$ corresponds to a closed $2$-chain $F_\lambda$ under the pairing in Eq.~(\ref{eq: cohomological pairing}).
\end{proposition}

\textbf{Group cohomology.} Cohomology enters at a second level, when a symmetry group acts on the Pauli observables and on the phase cocycle \cite{OkayRobertsBartlettRaussendorf,Raussendorf2016,OkayRaussendorf2020}. Suppose a group $G$ acts on the projective Pauli configuration and preserves the context incidence structure. For simplicity, we assume this group to correspond to a subgroup of the projective Clifford group, which is an extension of the group of symplectic transformations on the symplectic space $(V,\omega)$,
\begin{align*}
    1\longrightarrow V
    \longrightarrow \overline{\mathrm{Cliff}}_n
    \longrightarrow \Sp(2n,2)
    \longrightarrow 1\; .
\end{align*}
Its action on Pauli operators is defined up to sign only, which can be encoded in the phase function $\Phi:G\times\cV\ra\mathbb F_2$ which for every $g\in G$ with $\Phi_g(v):=\Phi(g,v)$ is given by 
\begin{align*}
    U_gT_vU_g^\dagger
    =(-1)^{\Phi_g(v)}T_{gv}\; .
\end{align*}
The functions $\Phi_g$ measure the difference between a projective symmetry and a strict symmetry of the chosen Pauli representatives.

\begin{proposition}\label{prop: phase function}
    With the notation above:
    
    \begin{enumerate}
        \item
        The family $\Phi=\{\Phi_g\}_{g\in G}$ is a group
        $1$-cocycle for the permutation action of $G$ on
        $\mathbb{F}_2^\cX$.
        
        \item For every $g\in G$,
        \begin{align}\label{eq: action on constraint vector}
            b(gC)-b(C)
            =(A\Phi_g)(C)\qquad(C\in\mathcal C)\; .
        \end{align}
        Consequently, $g b-b\in\im(A)$, hence, the contextuality class $[b]\in\operatorname{coker}A$ is $G$-invariant.
        
        \item If the sign vector itself is $G$-invariant, then $\Phi_g\in\ker A$ for $g\in G$, and thus
        \begin{align*}
            [\Phi]\in H^1(G,\ker A)\; .
        \end{align*}
        This class vanishes precisely when the observables can be rephased, without changing the context signs, so that the symmetry acts strictly: $U_gT'_vU_g^\dagger=T'_{gv}$.
    \end{enumerate}
\end{proposition}

\begin{proof}
    Compatibility with multiplication in $G$ gives, for $g,h\in G$,
    \begin{align*}
        U_hU_gT_vU_g^\dagger U_h^\dagger
        =(-1)^{\Phi_g(v)}U_hT_{gv}U_h^\dagger
        =(-1)^{\Phi_g(v)+\Phi_h(gv)}T_{hgv}\; .
    \end{align*}
    With the pullback convention $(g^*s)(v)=s(gv)$, this yields the $1$-cocycle identity $\Phi_{hg}=\Phi_g+g^*\Phi_h$.
    
    Now, since the operators in a context commute, conjugation by $U_g$ gives
    \begin{align*}
        (-1)^{b(C)}\one
        =U_g\left(\prod_{i\in C}T_v\right)U_g^\dagger
        =(-1)^{\sum_{v\in C}\Phi_g(v)}\prod_{v\in C}T_{gv}
        =(-1)^{\sum_{v\in C}\Phi_g(v)+b(gC)}\one\; ,
    \end{align*}
    proving Eq.~(\ref{eq: action on constraint vector}). In particular, since $gb-b$ is an incidence coboundary, $[b]\in\operatorname{coker}A$ is invariant. If $b$ is itself $G$-invariant, Eq.~(\ref{eq: action on constraint vector}) gives $A\Phi_g=0$, and the cocycle takes values in $\ker A$.
    
    Finally, under rephasing of the Hermitian Pauli representatives by $T'_v=(-1)^{s(v)}T_v$, we have
    \begin{align*}
        \Phi_g'(v)
        =\Phi_g(v)+s(v)+s(gv)\; ,
    \end{align*}
    hence, changing the Hermitian section changes $\Phi$ by a group $1$-coboundary. If $s\in\ker A$, this rephasing leaves every context sign unchanged. Consequently, a rephasing preserving $b$ makes all symmetry phases vanish exactly when $0=[\Phi]\in H^1(G,\ker A)$.
\end{proof}

There are therefore two distinct levels of symmetry in the case of quantum solutions built from $n$-qubit Pauli operators. The projective Clifford action is a symmetry of the finite symplectic compatibility geometry, whereas strict symmetry of the chosen Hermitian observables additionally requires trivialising the phase cocycle. No analogous extra datum occurs for the canonical projector realization $X_v=\zeta^{P_v}_2=\one-2P_v$, because projectors are phase-invariant.\\

\textbf{Projector reduction.} Unlike Pauli solutions, the projector constructions do not require a nontrivial cocycle, their obstruction is already present at the incidence level. In fact, since projectors are phase-invariant, no (phase) information is lost when reducing the arrangement to its incidence relations. While we leave open whether a nontrivial multiplication-phase cocycle analogous to the Pauli case can play such a role beyond the binary case, we point out that - beyond being equivalent descriptions of the same underlying contextual obstruction (see Prop.~\ref{prop: b <-> beta} and Prop.~\ref{prop: phase function}) - Pauli-based quantum solutions of binary LCS can be integrated into the language of projector-context arrangements.

For $p=2$, the canonical projector construction does not require the projectors to have rank one.  Let $C=\{P_1,\ldots,P_k\}$ be any orthogonal resolution of the identity on $\C^d$, that is, $P_iP_j=\delta_{ij}P_i$ and $\sum_{i=1}^k P_i=\one$. As in Def.~\ref{def: canonical LCS}, we associate to a projector-context arrangement the binary LCS
\begin{equation}\label{eq: projector-context LCS}
   \sum_{P\in C}x_P\equiv 1\pmod 2\qquad\forall C\in\cC\; .
\end{equation}

\begin{lemma}\label{lm: projective reduction}
    For every finite family of orthogonal resolutions of the identity, the associated LCS in Eq.~\eqref{eq: projector-context LCS} has the quantum solution $X_P=\one-2P$.
\end{lemma}

\begin{proof}
    Clearly $X_P^2=\one$, and projectors belonging to the same resolution commute. Since distinct projectors in a resolution are orthogonal, all mixed terms in the product vanish, hence,
    \begin{equation*}
       \prod_{P\in C}(\one-2P)
        =\one-2\sum_{P\in C}P
        =-\one\; .\qedhere
    \end{equation*}
\end{proof}

The familiar Pauli parity proofs thus admit projector refinements that fit the projector-context framework of Sec.~\ref{sec: incidence reduction}, for which classical unsatisfiability is decided by the existence of an incidence cycle via Prop.~\ref{prop: cycle criterion}, which further relates to the symmetry of the projector-context arrangement via Prop.~\ref{prop: orbit cycle}. We close by highlighting these aspects in two familiar examples.\\

\textbf{The Mermin--Peres square.} The Mermin--Peres square is usually presented as nine two-qubit Pauli observables arranged into six commuting triples, with an odd number of negative operator products \cite{Mermin,Peres}. Spectral refinement places the same phenomenon inside Peres's $24$ ray geometry. In particular, this geometry contains the well-known $18$ ray-$9$ context parity proof of Cabello, Estebaranz and Garc\'ia-Alcaine \cite{CabelloEstebaranzGarciaAlcaine,WaegellAravind2013}.

Let $\cP_\square=\{P_1,\ldots,P_{18}\}$ denote these rank-one projectors and $\cC_\square=\{C_1,\ldots,C_9\}$ the nine orthonormal tetrads. Every $P\in\cP_\square$ occurs in two members of $\cC_\square$. Hence, the canonical binary system $\sum_{P\in C}x_P=1\pmod 2$ for all $C\in\cC_\square$ has the quantum solution $X_P=\one-2P$ on $\C^4$.

The obstruction to a classical solution can be viewed in terms of symmetry. Regard the nine bases (not Pauli operators) as the vertices of a $3\times3$ grid and the eighteen projectors as its edges, where an edge joins the two bases in which the corresponding projector occurs. The incidence graph is the Cartesian product $K_3\square K_3$ of the complete graph on three vertices $K_3$, whose automorphism group is
\begin{align*}
   G_\square
   \cong (S_3\times S_3)\rtimes C_2\; ,\qquad |G_\square|=72\; .
\end{align*}
$G_\square$ is transitive both on contexts and projectors, hence, has a single projector and context orbit, and the complete orbit-incidence matrix is simply $N_\square=(2)=(0)$ over $\mathbb F_2$. Taking $\alpha=(1)$ gives
\begin{align*}
    N_\square\alpha\equiv 0\pmod 2\; ,\qquad
    |\,\cC_\square\,|\alpha
    =9\equiv 1\pmod 2\; .
\end{align*}
Classical unsatisfiability therefore follows from Prop.~\ref{prop: orbit cycle}; it explains the familiar parity argument without referring to Pauli multiplication: the invariant context chain $\lambda_\square=\sum_{C\in\cC_\square}[C]$ has
\begin{align*}
   A^T\lambda_\square\equiv 0\pmod 2\; ,\qquad
   \varepsilon(\lambda_\square)\equiv 1\pmod 2\; .
\end{align*}
Here, the Pauli signs of the original Mermin-Peres square have been traded for the uniform relation $b=\mathbf{1}_\cC$ attached to every orthogonal projector resolution as in Eq.~(\ref{eq: projector-context LCS}).

Comparing with Sec.~\ref{sec: HJS}, we also note for $n=2$, the HJS construction consists of $18$ rays and $9$ contexts, with every ray contained in exactly two contexts. Labelling the contexts by $B_{ij}$, the $u$-type rays join pairs of contexts within a common row and the $v$-type rays pairs within a common column, that is, the resulting incidence graph is $K_3\square K_3$, the same abstract $18$ ray-$9$-context incidence graph underlying the parity proof in Ref.~\cite{CabelloEstebaranzGarciaAlcaine} above.

Finally, the group $G_\square$ also relates to the geometry of two-qubit Paulis: up to phases, the nontrivial Pauli observables form the symplectic polar space $W(3,2)$ with automorphism group
\begin{align*}
   \Sp(4,2)\cong S_6\; ,\qquad |\Sp(4,2)|=720\; .
\end{align*}
There are ten Mermin squares in $W(3,2)$ \cite{SanigaPlanatPracnaHavlicek2007,HolweckBoutraySaniga2022}, whose stabiliser groups have order $\frac{720}{10}=72$, recovering $G_\square$. Thus the symmetry group of the projector incidence proof is the stabiliser of the corresponding Mermin configuration inside the ambient two-qubit symplectic geometry.\\

\textbf{Mermin's pentagram.} A similar analysis applies to Mermin's three-qubit pentagram. The pentagram consists of ten Pauli observables arranged into five commuting quadruples \cite{Mermin}. Abstractly, the five contexts may be identified with the vertices of $K_5$ and the ten observables with its edges: every observable belongs to the two contexts corresponding to its endpoints. The abstract automorphism group is therefore $G_\star\cong S_5$, it is the stabiliser of a Mermin pentagram in the ambient group $\Sp(6,2)$ for which the following exceptional relation holds \cite{LevayPlanatSaniga2013},
\begin{equation}\label{eq:Sp6E7}
    \Sp(6,2)\cong W(E_7)/\{\pm1\}\; .
\end{equation}
The Mermin pentagrams in $W(5,2)$ form a single $\Sp(6,2)$-orbit with stabiliser $S_5$, hence, there are
\begin{align*}
   \frac{|\Sp(6,2)|}{|S_5|}
    =\frac{1451520}{120}
    =12096
\end{align*}
Mermin pentagrams in total \cite{LevayPlanatSaniga2013,LevaySzabo2017,HolweckBoutraySaniga2022}.

For the projector formulation we use the standard refinement of the Kernaghan--Peres geometry associated with the pentagram \cite{KernaghanPeres1995}. Its five Pauli contexts generate forty joint eigenrays and a larger family of orthogonal bases \cite{WaegellAravind2013}. A particularly economical projector parity refinement consists of thirty rank-two projectors $\cP_\star=\{P_1,\ldots,P_{30}\}$, arranged into fifteen contexts $\cC_\star=\{C_1,\ldots,C_{15}\}$, where every context consists of four mutually orthogonal rank-two projectors resolving the identity on $\C^8$, and every projector occurs in exactly two contexts \cite{Toh2013}:
\begin{align*}
   \sum_{P\in C}x_P\equiv 1\pmod 2\qquad \forall C\in\cC_\star\; .
\end{align*}
Again, the associated LCS has the quantum solution $X_P=\one-2P$.

The parity contradiction again follows from Prop.~\ref{prop: orbit cycle}. Let $G$ be any subgroup of the automorphism group of this projector configuration and let $\cP_\star=\bigsqcup_i \cP_i$ and $\cC_\star=\bigsqcup_j O_j$ be the orbit decompositions. For the corresponding dual incidence matrix $N$, each row satisfies $\sum_jN_{ij}=2\equiv 0\pmod 2$ since every projector has total context degree two. Hence, with $\alpha=(1,\ldots,1)^T$, we find
\begin{align*}
    N\alpha&\equiv 0\pmod 2 &
    \sum_j|O_j|\alpha_j
    &=|\cC_\star|=15\equiv 1\pmod 2\; .
\end{align*}

Thus in both examples the projector formulation retains the all-context parity obstruction independently of how much of the symmetry of the parent Pauli configuration survives the spectral refinement. For the Mermin--Peres square the $18$-ray, $9$-context refinement retains a highly transitive grid symmetry. For Mermin's pentagram, the ten-observable Pauli configuration is $S_5$-symmetric inside $W(5,2)$, whereas a chosen $30$ projector-$15$ context refinement need not inherit the full $S_5$ action. This does not affect the obstruction: every projector has context degree two while the number of contexts is odd, hence, the all-context chain remains a nontrivial incidence cycle.

\end{document}